\documentclass[pra,aps,twocolumn,twoside,superscriptaddress]{revtex4}

\usepackage{amsmath,amsfonts,amssymb,caption,color,epsfig,graphics,graphicx,hyperref,latexsym,mathrsfs,revsymb,theorem,url,verbatim,enumerate,epstopdf,tikz,float,multirow,booktabs,appendix}
\hypersetup{colorlinks,linkcolor={blue},citecolor={red},urlcolor={blue}}
\usetikzlibrary{arrows, decorations.markings}

\tikzstyle{vecArrow} = [thick, decoration={markings,mark=at position
	1 with {\arrow[semithick]{open triangle 60}}},
double distance=1.4pt, shorten >= 5.5pt,
preaction = {decorate},
postaction = {draw,line width=1.4pt, white,shorten >= 4.5pt}]
\tikzstyle{innerWhite} = [semithick, white,line width=1.4pt, shorten >= 4.5pt]

\newtheorem{definition}{Definition}
\newtheorem{proposition}[definition]{Proposition}
\newtheorem{lemma}[definition]{Lemma}

\newtheorem{theorem}[definition]{Theorem}
\newtheorem{corollary}[definition]{Corollary}
\newtheorem{conjecture}[definition]{Conjecture}

\newtheorem{remark}[definition]{Remark}
\newtheorem{example}[definition]{Example}
\newtheorem{question}[definition]{Question}

\def\bcj{\begin{conjecture}}
	\def\ecj{\end{conjecture}}
\def\bcr{\begin{corollary}}
	\def\ecr{\end{corollary}}
\def\bd{\begin{definition}}
	\def\ed{\end{definition}}
\def\bea{\begin{eqnarray}}
\def\eea{\end{eqnarray}}
\def\bem{\begin{enumerate}}
	\def\eem{\end{enumerate}}
\def\bex{\begin{example}}
	\def\eex{\end{example}}
\def\bim{\begin{itemize}}
	\def\eim{\end{itemize}}
\def\bl{\begin{lemma}}
	\def\el{\end{lemma}}
\def\bma{\begin{bmatrix}}
	\def\ema{\end{bmatrix}}
\def\bpf{\begin{proof}}
	\def\epf{\end{proof}}
\def\bpp{\begin{proposition}}
	\def\epp{\end{proposition}}
\def\bqu{\begin{question}}
	\def\equ{\end{question}}
\def\br{\begin{remark}}
	\def\er{\end{remark}}
\def\bt{\begin{theorem}}
	\def\et{\end{theorem}}

\def\squareforqed{\hbox{\rlap{$\sqcap$}$\sqcup$}}
\def\qed{\ifmmode\squareforqed\else{\unskip\nobreak\hfil
		\penalty50\hskip1em\null\nobreak\hfil\squareforqed
		\parfillskip=0pt\finalhyphendemerits=0\endgraf}\fi}
\def\endenv{\ifmmode\;\else{\unskip\nobreak\hfil
		\penalty50\hskip1em\null\nobreak\hfil\;
		\parfillskip=0pt\finalhyphendemerits=0\endgraf}\fi}
\newenvironment{proof}{\noindent \textbf{{Proof.~} }}{\qed}
\def\Dbar{\leavevmode\lower.6ex\hbox to 0pt
	{\hskip-.23ex\accent"16\hss}D}
\makeatletter
\def\url@leostyle{%
	\@ifundefined{selectfont}{\def\UrlFont{\sf}}{\def\UrlFont{\small\ttfamily}}}
\makeatother
\def\bcj{\begin{conjecture}}
	\def\ecj{\end{conjecture}}
\def\bcr{\begin{corollary}}
	\def\ecr{\end{corollary}}
\def\bd{\begin{definition}}
	\def\ed{\end{definition}}
\def\bea{\begin{eqnarray}}
\def\eea{\end{eqnarray}}
\def\bem{\begin{enumerate}}
	\def\eem{\end{enumerate}}
\def\bex{\begin{example}}
	\def\eex{\end{example}}
\def\bim{\begin{itemize}}
	\def\eim{\end{itemize}}
\def\bl{\begin{lemma}}
	\def\el{\end{lemma}}
\def\bpf{\begin{proof}}
	\def\epf{\end{proof}}
\def\bpp{\begin{proposition}}
	\def\epp{\end{proposition}}
\def\bqu{\begin{question}}
	\def\equ{\end{question}}
\def\br{\begin{remark}}
	\def\er{\end{remark}}
\def\bt{\begin{theorem}}
	\def\et{\end{theorem}}

\def\btb{\begin{tabular}}
	\def\etb{\end{tabular}}

\newcommand{\nc}{\newcommand}

\nc{\bbA}{\mathbb{A}} \nc{\bbB}{\mathbb{B}} \nc{\bbC}{\mathbb{C}}
\nc{\bbD}{\mathbb{D}} \nc{\bbE}{\mathbb{E}} \nc{\bbF}{\mathbb{F}}
\nc{\bbG}{\mathbb{G}} \nc{\bbH}{\mathbb{H}} \nc{\bbI}{\mathbb{I}}
\nc{\bbJ}{\mathbb{J}} \nc{\bbK}{\mathbb{K}} \nc{\bbL}{\mathbb{L}}
\nc{\bbM}{\mathbb{M}} \nc{\bbN}{\mathbb{N}} \nc{\bbO}{\mathbb{O}}
\nc{\bbP}{\mathbb{P}} \nc{\bbQ}{\mathbb{Q}} \nc{\bbR}{\mathbb{R}}
\nc{\bbS}{\mathbb{S}} \nc{\bbT}{\mathbb{T}} \nc{\bbU}{\mathbb{U}}
\nc{\bbV}{\mathbb{V}} \nc{\bbW}{\mathbb{W}} \nc{\bbX}{\mathbb{X}}
\nc{\bbZ}{\mathbb{Z}}

\nc{\bA}{{\bf A}} \nc{\bB}{{\bf B}} \nc{\bC}{{\bf C}}
\nc{\bD}{{\bf D}} \nc{\bE}{{\bf E}} \nc{\bF}{{\bf F}}
\nc{\bG}{{\bf G}} \nc{\bH}{{\bf H}} \nc{\bI}{{\bf I}}
\nc{\bJ}{{\bf J}} \nc{\bK}{{\bf K}} \nc{\bL}{{\bf L}}
\nc{\bM}{{\bf M}} \nc{\bN}{{\bf N}} \nc{\bO}{{\bf O}}
\nc{\bP}{{\bf P}} \nc{\bQ}{{\bf Q}} \nc{\bR}{{\bf R}}
\nc{\bS}{{\bf S}} \nc{\bT}{{\bf T}} \nc{\bU}{{\bf U}}
\nc{\bV}{{\bf V}} \nc{\bW}{{\bf W}} \nc{\bX}{{\bf X}}
\nc{\bZ}{{\bf Z}} \nc{\bm}{{\bf m}} \nc{\bv}{{\bf v}}
\nc{\ba}{{\bf a}} \nc{\be}{{\bf e}} \nc{\bu}{{\bf u}}
\nc{\brr}{{\bf r}}

\nc{\cA}{{\cal A}} \nc{\cB}{{\cal B}} \nc{\cC}{{\cal C}}
\nc{\cD}{{\cal D}} \nc{\cE}{{\cal E}} \nc{\cF}{{\cal F}}
\nc{\cG}{{\cal G}} \nc{\cH}{{\cal H}} \nc{\cI}{{\cal I}}
\nc{\cJ}{{\cal J}} \nc{\cK}{{\cal K}} \nc{\cL}{{\cal L}}
\nc{\cM}{{\cal M}} \nc{\cN}{{\cal N}} \nc{\cO}{{\cal O}}
\nc{\cP}{{\cal P}} \nc{\cQ}{{\cal Q}} \nc{\cR}{{\cal R}}
\nc{\cS}{{\cal S}} \nc{\cT}{{\cal T}} \nc{\cU}{{\cal U}}
\nc{\cV}{{\cal V}} \nc{\cW}{{\cal W}} \nc{\cX}{{\cal X}}
\nc{\cZ}{{\cal Z}}

\nc{\hA}{{\hat{A}}} \nc{\hB}{{\hat{B}}} \nc{\hC}{{\hat{C}}}
\nc{\hD}{{\hat{D}}} \nc{\hE}{{\hat{E}}} \nc{\hF}{{\hat{F}}}
\nc{\hG}{{\hat{G}}} \nc{\hH}{{\hat{H}}} \nc{\hI}{{\hat{I}}}
\nc{\hJ}{{\hat{J}}} \nc{\hK}{{\hat{K}}} \nc{\hL}{{\hat{L}}}
\nc{\hM}{{\hat{M}}} \nc{\hN}{{\hat{N}}} \nc{\hO}{{\hat{O}}}
\nc{\hP}{{\hat{P}}} \nc{\hR}{{\hat{R}}} \nc{\hS}{{\hat{S}}}
\nc{\hT}{{\hat{T}}} \nc{\hU}{{\hat{U}}} \nc{\hV}{{\hat{V}}}
\nc{\hW}{{\hat{W}}} \nc{\hX}{{\hat{X}}} \nc{\hZ}{{\hat{Z}}}

\nc{\hn}{{\hat{n}}}

\def\ghz{\mathop{\rm GHZ}}

\def\wi{\widetilde}

\newcommand{\bra}[1]{\langle#1|}
\newcommand{\ket}[1]{|#1\rangle}

\newcommand{\ketbra}[2]{|#1\rangle\!\langle#2|}

\def\Dbar{\leavevmode\lower.6ex\hbox to 0pt
	{\hskip-.23ex\accent"16\hss}D}

\begin{document}
	
\title{Strong quantum nonlocality with entanglement}
	
	

\author{Fei Shi}
\email[]{shifei@mail.ustc.edu.cn}
\affiliation{School of Cyber Security,
		University of Science and Technology of China, Hefei, 230026, People's Republic of China.}
	
\author{Mengyao Hu}\email[]{mengyaohu@buaa.edu.cn }
\affiliation{School of Mathematical Sciences, Beihang University, Beijing 100191, China}
	
\author{Lin Chen}
\email[]{linchen@buaa.edu.cn}
\affiliation{School of Mathematical Sciences, Beihang University, Beijing 100191, China}
\affiliation{International Research Institute for Multidisciplinary Science, Beihang University, Beijing 100191, China}
	
\author{Xiande Zhang}
\email[]{drzhangx@ustc.edu.cn}
\affiliation{School of Mathematical Sciences,
		University of Science and Technology of China, Hefei, 230026, People's Republic of China}

\begin{abstract}
Strong quantum nonlocality was introduced recently as a stronger manifestation of nonlocality in multipartite systems through the notion of local irreducibility in all bipartitions. Known existence results for sets of strongly nonlocal orthogonal states are limited to product states. In this paper, based on the Rubik's cube, we give the first construction of such sets consisting of entangled states in  $d\otimes d\otimes d$ for all $d\geq 3$. Consequently, we answer an open problem given by  Halder \emph{et al.} [\href{https://journals.aps.org/prl/abstract/10.1103/PhysRevLett.122.040403}{Phys. Rev. Lett. \textbf{122}, 040403 (2019)}], that is, orthogonal entangled bases that are strongly nonlocal do exist. Furthermore, we propose two entanglement-assisted protocols for local discrimination of our results.  Each  protocol consumes less entanglement resource than the teleportation-based protocol averagely. Our results exhibit the phenomenon of strong quantum nonlocality with entanglement.

\end{abstract}
	
\maketitle
	
	
	
\section{Introduction}\label{sec:int}
A set of orthogonal quantum states is locally indistinguishable, if it is not possible to optimally
distinguish the states by any sequence of local operations and classical communications (LOCC). It exhibits the phenomenon  of quantum nonlocality.  Local indistinguishability can be used for data
hiding \cite{terhal2001hiding,divincenzo2002quantum,eggeling2002hiding,Matthews2009Distinguishability} and quantum secret sharing \cite{Markham2008Graph}.  Any three Bell states cannot be locally distinguished \cite{ghosh2001distinguishability}. The phenomenon of more nonlocality with less entanglement was shown in Ref.~\cite{PhysRevLett.90.047902}.
Bennett \emph{et al.} first constructed a locally indistinguishable orthogonal product basis in bipartite Hilbert space $ 3\otimes 3$, which shows the phenomenon of quantum nonlocality without entanglement \cite{bennett1999quantum}. Later, locally indistinguishable orthogonal entangled sets  and orthogonal product sets are widely investigated \cite{1,2,3,4,5,6,7,8,9,10,11,12,13,14,15,16,17,18}.
	
Recently,  Halder \emph{et al.} proposed the concept of \emph{locally irreducible set} \cite{Halder2019Strong}.  It is a set of orthogonal quantum states that it is impossible to locally eliminate one or more
states from the set by orthogonality-preserving local measurements. Local irreducibility sufficiently ensures local indistinguishability,  while the converse is not true.  In $ 3\otimes 3\otimes 3$ and $ 4\otimes 4\otimes 4$, they constructed two orthogonal product bases  that are locally irreducible in all bipartitions.  It shows the phenomenon of strong quantum nonlocality without entanglement.  Ref.~\cite{yuan2020strong} constructed the strongly nonlocal orthogonal product sets (SNOPSs) of size $6(d^2-1)$ in $ d\otimes d\otimes d$ for $d\geq 3$, and a strongly nonlocal orthogonal product basis (SNOPB) in $ 3\otimes 3\otimes 3\otimes 3$. Ref.~\cite{PhysRevA.99.062108} generalized the definition of strong nonlocality based on the local irreducibility in some multipartitions,  and gave some examples in $ 3\otimes 3\otimes 3$ and   $ 3\otimes 3\otimes 3\otimes 3$. In spite of these constructions, the existence of orthogonal entangled sets that are locally irreducible in all bipartitions remains unknown. An open question has been proposed to find 
orthogonal entangled bases that are locally irreducible in all bipartitions \cite{Halder2019Strong}. Such bases are called strongly nonlocal orthogonal entangled bases (SNOEBs). We shall give a positive answer to this open question.

In this paper, we construct strongly nonlocal orthogonal entangled sets (SNOESs) and SNOEBs in $ d\otimes d\otimes d$ for $d\geq 3$, and provide two efficient entanglement-assisted discrimination protocols for an SNOES in $ 3\otimes 3\otimes 3$.  
First, by using Fig.~\ref{Figure:tilebell} and Fig.~\ref{Figure:tite333}, we construct an orthogonal entangled set of size $24$ and an orthogonal entangled basis in $ 3\otimes 3\otimes 3$, and we prove these two sets are both strongly nonlocal by using  Fig.~\ref{Figure:tite39} in Lemma~\ref{lem:strong}.  Then, we show an SNOES of size $54$ and an SNOEB in  $ 4\otimes 4\otimes 4$ in Lemma~\ref{lem:strong444}. Next, by using Lemma~\ref{lem:strong}, Lemma~\ref{lem:strong444} and Fig.~\ref{Figure:555666}, we show  an SNOES of size $d^3-d$ and an SNOEB in $ d\otimes d\otimes d$ when $d\geq 3$ is odd,  and an SNOES of size $d^3-d-6$ and an SNOEB in $ d\otimes d\otimes d$ when $d\geq 3$ is even in Theorem~\ref{thm:generalddd}.   
Finally, we give two entanglement-assisted discrimination protocols for the SNOES in $ 3\otimes 3\otimes 3$ in Proposition~\ref{pro:distinguish39} and Proposition~\ref{pro:distinguish333}. Each  protocol consumes less entanglement resource than the teleportation-based protocol averagely.

Entanglement-assisted discrimination also attracts more and more attention \cite{ghosh2001distinguishability,cohen2008understanding,bandyopadhyay2016entanglement,zhang2016entanglement,gungor2016entanglement,zhang2018local,Sumit2019Genuinely,zhang2020locally,Shi2020Unextendible}. By using sufficient entanglement, a set of orthogonal states  can be always distinguished through the teleportation-based protocol \cite{Bennett1993Teleporting}.  Since entanglement is a costly resource, the discrimination with less entanglement is desirable. It is known that unextendible product bases (UPBs) can not be locally distinguished \cite{de2004distinguishability}.  A two-qutrit UPB of size five can be locally distinguished with a two-qubit maximally entangled state \cite{cohen2008understanding}.  Since a strongly nonlocal orthogonal set cannot be locally distinguished in every bipartition, a perfect local discrimination of this set would require a resource state that must be entangled in all bipartitions.  Ref.~\cite{Sumit2019Genuinely} gave different entanglement-assisted  discrimination protocols for some SNOPBs, and each protocol consumes less entanglement resource than the teleportation-based protocol averagely.  By comparing our entanglement-assisted discrimination protocols for the SNOES with those for the SNOPB in $3\otimes 3\otimes 3$, we show that the  entanglement can increase the difficulty to locally distinguish orthogonal states.
	
The rest of this paper is organized as follows. In Sec.~\ref{sec:prelimi},  we introduce the preliminary knowledge used in this paper. In Sec.~\ref{sec:strong333}, we give an elegant construction of SNOESs and SNOEBs in $ d\otimes d\otimes d$ for $d\geq 3$ by using a $d\times d\times d$ Rubik's cube. In Sec.~\ref{sec:discri-entang}, we investigate the entanglement-assisted discrimination protocols for the SNOES in $ 3\otimes 3\otimes 3$. Finally, we conclude in Sec.~\ref{sec:con}. 
	
\section{Preliminary}\label{sec:prelimi}
	
Throughout this paper, we do not normalize states
and operators for simplicity, and we consider only pure states and POVM measurements. A set of orthogonal states is \emph{locally indistinguishable}, if it is not possible to distinguish the states by any sequence of
local operations and classical communications (LOCC). A measurement performed to distinguish a set of mutually orthogonal  states is called an \emph{orthogonality-preserving measurement} if after the measurement the states
remain mutually orthogonal. Further, a measurement is nontrivial if not all the POVM elements are proportional to the identity operator. Otherwise, the measurement is trivial.
		
Consider an $n$-partite quantum system  with Hilbert space $ {d_1}\otimes {d_2}\otimes\cdots\otimes  {d_n}$.  A set of orthogonal quantum states is called a  \emph{locally irreducible set} if it is not possible to eliminate one or more states from the set by
orthogonality-preserving local measurements \cite{Halder2019Strong}. The idea is to check whether an orthogonality-preserving POVM on any of the subsystems is trivial or not. If it is trivial for all subsystems, then the set of states is locally irreducible.
	
Obviously, local irreducibility sufficiently ensures local indistinguishability. However, the converse is not true. For example, consider the following set in $ 2\otimes 3$,
\begin{equation}\label{eq:fivebell}
\begin{split}
&\ket{\psi_{1,2}}=\ket{0,0}\pm\ket{1,1}, \quad \ket{\psi_{3,4}}=\ket{0,1}\pm\ket{1,0},\\
&\ket{\psi_5}=\ket{0,2},
\end{split}
\end{equation}
where $\ket{\psi_{j_1,j_2}}=\ket{k_1,k_2}\pm\ket{k_3,k_4}$ means $\ket{\psi_{j_1}}=\ket{k_1,k_2}+\ket{k_3,k_4}$ and $\ket{\psi_{j_2}}=\ket{k_1,k_2}-\ket{k_3,k_4}$. Since the Bell basis can not be locally distinguished \cite{ghosh2001distinguishability}, $\{\ket{\psi_i}\}_{i=1}^5$ given by Eq.~(\ref{eq:fivebell}) is locally indistinguishable. However, $\{\ket{\psi_i}\}_{i=1}^5$ is locally reducible, since Bob can use the measurement $\{\ketbra{2}{2}, I-\ketbra{2}{2}\}$ to eliminate $\{\ket{\psi_i}\}_{i=1}^4$ and $\ket{\psi_5}$, respectively. If we only consider the Bell basis in $ 2\otimes 2$,
\begin{equation}\label{eq:bell}
\ket{\psi_{1,2}}=\ket{0,0}\pm\ket{1,1}, \quad \ket{\psi_{3,4}}=\ket{0,1}\pm\ket{1,0},
\end{equation}
it is locally irreducible \cite{Halder2019Strong}. Let Alice go first and start the orthogonality-preserving POVM $E_m=M^\dagger M$. Each POVM element can be written as a $2\times 2$ matrix in the basis $\{\ket{0},\ket{1}\}$: $E_m=\begin{pmatrix}a_{0,0} &a_{0,1}\\ a_{1,0} &a_{1,1}\end{pmatrix}.$
Then the postmeasurement states of $\{M\otimes I\ket{\psi_k}\}_{k=1}^4$ should be mutually orthogonal. Since $\bra{\psi_1}M^\dagger M\otimes I\ket{\psi_2}=0$, it implies  $a_{0,0}=a_{1,1}$. Moreover, since $\bra{\psi_1}M^\dagger M\otimes I\ket{\psi_3}=\bra{\psi_1}M^\dagger M\otimes I\ket{\psi_4}=0$, we obtain $a_{0,1}\pm a_{1,0}=0$. It implies $a_{0,1}= a_{1,0}=0$. Then $E_{m}$ is trivial. It means that Alice cannot go first.  Bob also cannot go first from the symmetry of the Bell basis. Thus, the Bell basis is locally irreducible.
	
In Ref.~\cite{Halder2019Strong}, the authors considered strong quantum nonlocality without entanglement. Although they only defined for product states, it is natural to extend it for general orthogonal states. In $ {d_1}\otimes {d_2}\otimes\cdots\otimes  {d_n}$, $n\geq 3$, a set of orthogonal states is \emph{strongly nonlocal} if it is locally irreducible in every bipartition.

 There exists a locally irreducible set that is not strongly nonlocal \cite{Halder2019Strong}. For example, three-qubit $\ghz$ basis, $\ket{\varphi_{1,2}}=\ket{0,0,0}\pm\ket{1,1,1}$, $\ket{\varphi_{3,4}}=\ket{0,1,1}\pm\ket{1,0,0}$, $\ket{\varphi_{5,6}}=\ket{0,0,1}\pm\ket{1,1,0}$, $\ket{\varphi_{7,8}}=\ket{0,1,0}\pm\ket{1,0,1}$ is locally irreducible. If we consider $A|BC$ bipartition, this basis is  locally reducible. Since Bob and Charlie can use the the measurement $\{\ketbra{0,0}{0,0}+\ketbra{1,1}{1,1}, I-(\ketbra{0,0}{0,0}+\ketbra{1,1}{1,1})\}$ to eliminate $\{\ket{\varphi_i}\}_{i=5}^8$ and $\{\ket{\varphi_i}\}_{i=1}^4$, respectively.  The authors in Ref.~\cite{Halder2019Strong}  proposed an open question, whether one can find orthogonal entangled  bases  that are locally irreducible in all bipartitions. That is to find strongly nonlocal orthogonal entangled bases (SNOEBs). Their intuition is that a genuinely entangled orthogonal basis
(the basis vectors are entangled in every bipartition) might be a promising candidate. However,  they showed that the $N$-qubit $\ghz$ basis is locally reducible in all bipartitions. We have also tried some genuinely entangled orthogonal bases in Refs.~ \cite{li2019k,shishenchenzhang,raissi2018optimal,raissi2019constructing}, but they are not strongly nonlocal. Thus, we begin to consider an orthogonal entangled basis which contains entangled states that are not genuinely entangled.  More generally, we will consider  strongly nonlocal orthogonal entangled sets (SNOESs) which do not form a complete basis. 
		
\section{SNOESs and SNOEBs in $ d\otimes d\otimes d$ for $d\geq 3$ }\label{sec:strong333}
	
In this section, we give an elegant construction of an SNOES and an SNOEB in $ 3\otimes 3\otimes 3$ in Lemma~\ref{lem:strong}. Similarly, we show an SNOES and an SNOEB in $ 4\otimes  4\otimes  4$ in Lemma~\ref{lem:strong444}. Further, we generalize these two constructions to $ d\otimes d\otimes d$ for any $d\geq 3$ in Theorem~\ref{thm:generalddd}.
	
\begin{figure}[b]
	\centering
	\includegraphics[scale=0.45]{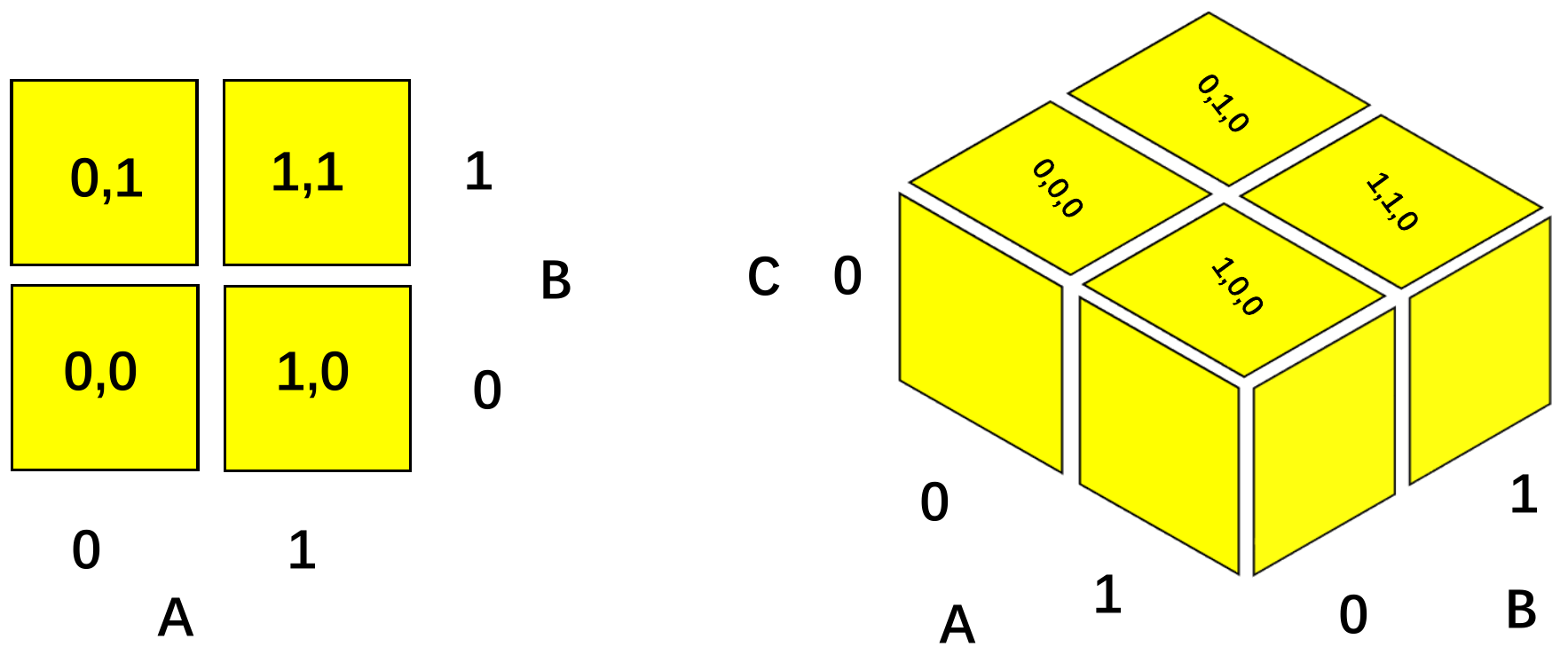}
	\caption{Bell states in the  $2\times 2$ grid.}  \label{Figure:tilebell}
\end{figure}

\subsection{An SNOES and an SNOEB in $ 3\otimes 3\otimes 3$}		
The left figure of Fig.~\ref{Figure:tilebell} is a $2\times 2$ grid. If we choose the diagonal grid cells, we can construct four states, $\ket{0,0}\pm\ket{1,1}$ and $\ket{0,1}\pm\ket{1,0}$. Obviously, these four states form the Bell basis in $ 2\otimes  2$. See also  Eq.~(\ref{eq:bell}). If we add an ancillary system $C$ (see the right figure of Fig.~\ref{Figure:tilebell}), these four states are transformed into $\ket{0,0,0}\pm\ket{1,1,0}$ and $\ket{0,1,0}\pm\ket{1,0,0}$. Each state is an entangled state across bipartitions $A|BC$ and $B|AC$, and it is a product state across $C|AB$ bipartition. Thus, it is an entangled state, but it is not a genuinely entangled state.

\begin{figure}[t]
	\centering
	\includegraphics[scale=0.5]{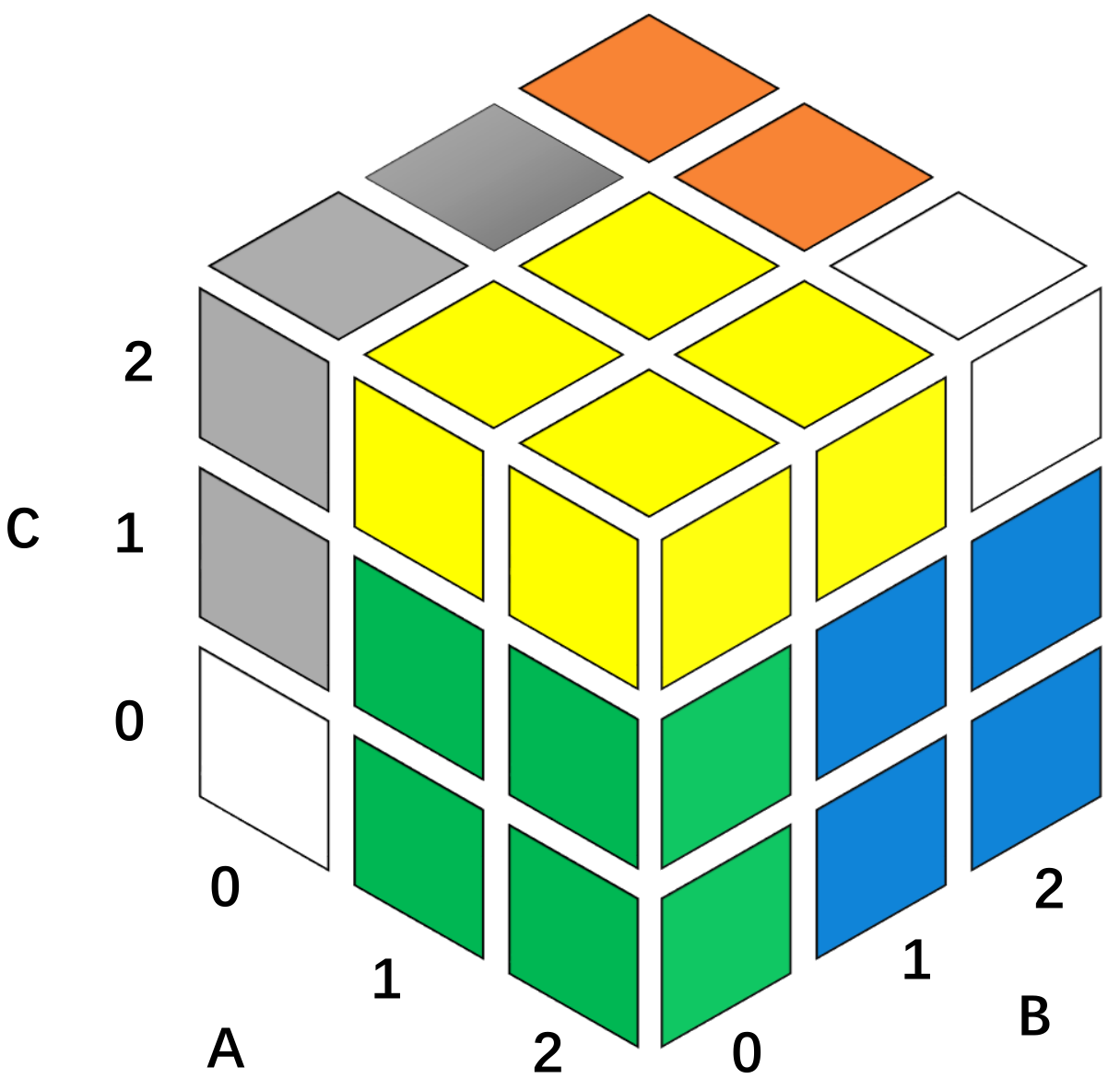}
	\caption{$3\times 3\times 3$
		Rubik's cube with six $2\times 2\times 1$ subcubes, $\{1,2\}\times \{0\}\times\{0,1\}$,  $\{1,2\}\times \{0,1\}\times\{2\}$,  $\{2\}\times \{1, 2\}\times\{0,1\}$, $\{0,1\}\times \{1,2\}\times\{0\}$,   $\{0,1\}\times \{2\}\times\{1,2\}$,  $\{0\}\times \{0,1\}\times\{1,2\}$, and three  $1\times 1\times 1$ subcubes，  $\{0\}\times\{0\}\times\{0\}$, $\{1\}\times\{1\}\times\{1\}$,  $\{2\}\times\{2\}\times\{2\}$. Note that $J_1 \times J_2\times J_3$ means that the index set of $A$ part is $J_1$, $B$ part is $J_2$, and $C$ part is $J_3$.  }  \label{Figure:tite333}
\end{figure}
In Fig.~\ref{Figure:tite333}, there are six $2\times 2\times 1$ subcubes, $\{1,2\}\times \{0\}\times\{0,1\}$,  $\{1,2\}\times \{0,1\}\times\{2\}$,  $\{2\}\times \{1, 2\}\times\{0,1\}$, $\{0,1\}\times \{1,2\}\times\{0\}$,   $\{0,1\}\times \{2\}\times\{1,2\}$,  $\{0\}\times \{0,1\}\times\{1,2\}$, where $J_1 \times J_2\times J_3$ means that the index set of $A$ part is $J_1$, $B$ part is $J_2$, and $C$ part is $J_3$. From the point of Fig.~\ref{Figure:tilebell}, we can construct an entangled set in $ 3\otimes 3\otimes 3$ by the six $2\times 2\times 1$ subcubes,
\begin{equation}\label{eq:tile333}
\begin{aligned}
\ket{\psi_{1,2}}&=\ket{1,0,0}\pm\ket{2,0,1}, \;  \, \     \ket{\psi_{3,4}}=\ket{1,0,1}\pm\ket{2,0,0},    \\
\ket{\psi_{5,6}}&=\ket{1,0,2}\pm\ket{2,1,2}, \;  \, \ \ket{\psi_{7,8}}=\ket{1,1,2}\pm\ket{2,0,2},  \\	\ket{\psi_{9,10}}&=\ket{2,1,0}\pm\ket{2,2,1},    \ket{\psi_{11,12}}=\ket{2,1,1}\pm\ket{2,2,0},    \\ \ket{\psi_{13,14}}&=\ket{0,1,0}\pm\ket{1,2,0},  \ket{\psi_{15,16}}=\ket{0,2,0}\pm\ket{1,1,0}, \\
\ket{\psi_{17,18}}&=\ket{0,2,1}\pm\ket{1,2,2},  \ket{\psi_{19,20}}=\ket{0,2,2}\pm\ket{1,2,1},\\
\ket{\psi_{21,22}}&=\ket{0,0,1}\pm\ket{0,1,2},  \ket{\psi_{23,24}}=\ket{0,0,2}\pm\ket{0,1,1}.  
\end{aligned}
\end{equation}	
 The states of $\{\ket{\psi_k}\}_{k=1}^{24}$  must be mutually orthogonal due to the disjointness of subcubes in Fig.~\ref{Figure:tite333}.  For each state $\ket{\psi_k}$, $1\leq k\leq 24$, there must exist two bipartitions of $\{A|BC, B|AC, C|AB\}$ such that $\ket{\psi_k}$ is an entangled state, while it is a product state across the remaining bipartition.
Next, we  extend this orthogonal entangled set $\{\ket{\psi_k}\}_{k=1}^{24}$ to an orthogonal entangled basis. Since there are three  $1\times 1\times 1$ subcubes,  $\{0\}\times\{0\}\times\{0\}$, $\{1\}\times\{1\}\times\{1\}$,  $\{2\}\times\{2\}\times\{2\}$ left in Fig.~\ref{Figure:tite333}, we can choose three $\ghz$ states:
\begin{equation}\label{eq:threeghz}
\begin{split}
&\ket{\psi_{25}}=\ket{0,0,0}+\ket{1,1,1}+\ket{2,2,2}, \\
&\ket{\psi_{26}}=\ket{0,0,0}+w_3\ket{1,1,1}+w_3^2\ket{2,2,2}, \\ 
&\ket{\psi_{27}}=\ket{0,0,0}+w_3^2\ket{1,1,1}+w_3\ket{2,2,2}.
\end{split}
\end{equation}
Obviously, these three $\ghz$ states are mutually orthogonal, and they are all  genuinely entangled states.
Then $\{\ket{\psi_k}\}_{k=1}^{27}$ given by Eqs.~(\ref{eq:tile333}) and (\ref{eq:threeghz}) forms an orthogonal entangled basis in $ 3\otimes 3\otimes 3$  by the structure of Fig.~\ref{Figure:tite333}. In the following, we show that the orthogonal entangled set $\{\ket{\psi_k}\}_{k=1}^{24}$ and the corresponding basis $\{\ket{\psi_k}\}_{k=1}^{27}$  are both strongly nonlocal.
	
\begin{lemma}\label{lem:strong}
In $ 3\otimes 3\otimes 3$, the orthogonal entangled set $\{\ket{\psi_k}\}_{k=1}^{24}$ given by Eq.~(\ref{eq:tile333}) is strongly nonlocal. The orthogonal entangled basis $\{\ket{\psi_k}\}_{k=1}^{27}$ given by Eqs.~(\ref{eq:tile333}) and (\ref{eq:threeghz}) is also strongly nonlocal.
\end{lemma}
\begin{proof}
Since the strong nonlocality of $\{\ket{\psi_k}\}_{k=1}^{24}$ can imply the strong nonlocality of $\{\ket{\psi_k}\}_{k=1}^{27}$, we only need to show that $\{\ket{\psi_k}\}_{k=1}^{24}$ is strongly nonlocal.
		
First, we consider $A|BC$ bipartition. Define a bijection from the basis  $\{\ket{p,q}\}_{p,q=0}^{2}$ in $ 3\otimes 3$ to the  basis in $\bbC^9$ as follows: $\ket{0,0}\rightarrow\ket{0}$, $\ket{0,1}\rightarrow\ket{1}$, $\ket{0,2}\rightarrow\ket{2}$, $\ket{1,0}\rightarrow\ket{5}$, $\ket{1,1}\rightarrow\ket{4}$,  $\ket{1,2}\rightarrow\ket{3}$, $\ket{2,0}\rightarrow\ket{6}$, $\ket{2,1}\rightarrow\ket{7}$, $\ket{2,2}\rightarrow\ket{8}$. Then we rewrite the set of states $\{\ket{\psi_k}\}_{k=1}^{24}$ in $ 3\otimes 3\otimes 3$  as $\{\ket{\varphi_k}\}_{k=1}^{24}$ in $ 3\otimes 9$,
\begin{equation}\label{eq:tile39}
\begin{aligned}
\ket{\varphi_{1,2}}&=\ket{1,0}\pm\ket{2,1},\ &\ket{\varphi_{3,4}}&=\ket{1,1}\pm\ket{2,0},\\
\ket{\varphi_{5,6}}&=\ket{1,2}\pm\ket{2,3},\ &\ket{\varphi_{7,8}}&=\ket{1,3}\pm\ket{2,2},\\
\ket{\varphi_{9,10}}&=\ket{2,5}\pm\ket{2,7},\ &\ket{\varphi_{11,12}}&=\ket{2,4}\pm\ket{2,6},\\
\ket{\varphi_{13,14}}&=\ket{0,5}\pm\ket{1,6},\ &\ket{\varphi_{15,16}}&=\ket{0,6}\pm\ket{1,5},\\
\ket{\varphi_{17,18}}&=\ket{0,7}\pm\ket{1,8},\ &\ket{\varphi_{19,20}}&=\ket{0,8}\pm\ket{1,7},\\
\ket{\varphi_{21,22}}&=\ket{0,1}\pm\ket{0,3},\ &\ket{\varphi_{23,24}}&=\ket{0,2}\pm\ket{0,4}.\\
\end{aligned}
\end{equation}
Eq.~(\ref{eq:tile39}) corresponds to the $3\times 9$ grid in Fig.~\ref{Figure:tite39}. Every gird has an index $(i,j)$, where $i$ is the row index of $A$ part, and $j$ is the column index of $BC$ part. For example, $\ket{\varphi_{1,2}}$ corresponds to the cell set $\{(1,0),(2,1)\}$. We need to show  that $\{\ket{\varphi_k}\}_{k=1}^{24}$ given by Eq.~(\ref{eq:tile39}) is locally irreducible.

\begin{figure}[t]
\centering
\includegraphics[scale=0.73]{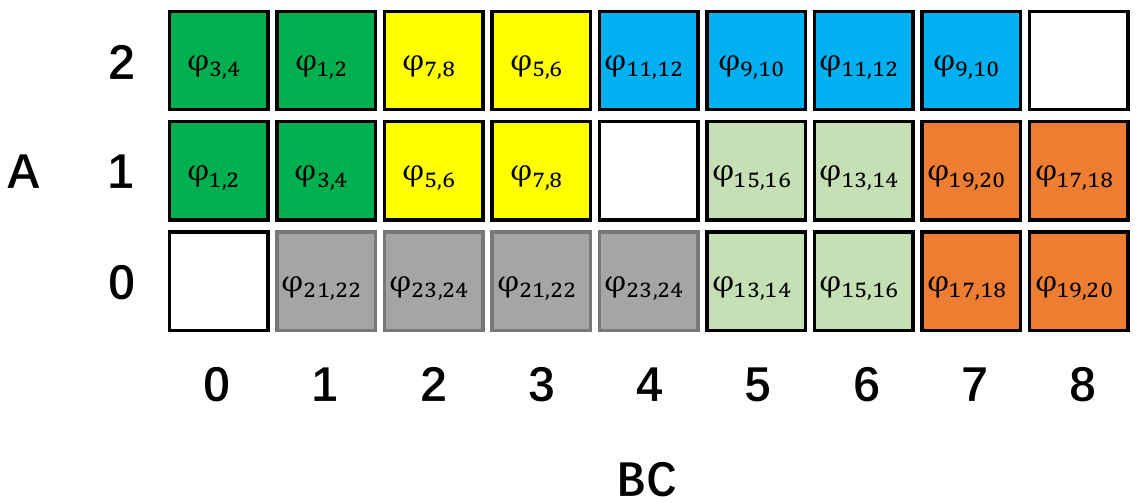}
\caption{The corresponding $3\times 9$ grid of Eq.~(\ref{eq:tile39}).   Every gird has an index $(i,j)$, where $i$ is the row index in $A$ part, and $j$ is the column index in $BC$ part. For example, $\ket{\varphi_{1,2}}$ corresponds to the cell set $\{(1,0),(2,1)\}$. } \label{Figure:tite39}
\end{figure}
		
Let Alice go first and start the orthogonality-preserving POVM,
\begin{equation}
E_m=M_1^\dagger M_1=
\begin{pmatrix}
a_{0,0} & a_{0,1}  &a_{0,2}\\
a_{1,0} & a_{1,1}  &a_{1,2}\\
a_{2,0} & a_{2,1}  &a_{2,2}\\	
\end{pmatrix}.
\end{equation}
Then the states of $\{M_1\otimes I\ket{\varphi_{k}}\}_{k=1}^{24}$ are mutually orthogonal.  In order to show that the off-diagonal elements of $E_m$ are all zero, we need to choose the cells with same column index in Fig.~\ref{Figure:tite39}. For example, if we choose the cell sets $\{(0,1),(0,3)\}$ and $\{(1,1),(2,0)\}$, then it means that we choose the states $\ket{\varphi_{21,22}}$ and $\ket{\varphi_{3,4}}$. Then $\bra{\varphi_{21}} E_m\otimes I\ket{\varphi_{3}}=(\bra{0,1}+\bra{0,3})E_m\otimes I(\ket{1,1}+\ket{2,0}=\bra{0,1}E_m\otimes I_9\ket{1,1}=\bra{0}E_m\ket{1}=a_{0,1}=0$, where $0$ and $1$ are the row indices of cells $(0,1)$ and $(1,1)$  respectively.  In this way, we obtain $a_{0,2}=0$ by $\{(0,1),(0,3)\}$ and $\{(2,1),(1,0)\}$. If we choose $\{(1,0),(2,1)\}$ and $\{(2,0),(1,1)\}$, then $\bra{\varphi_{1}} E_m\otimes I\ket{\varphi_{3}}=\bra{\varphi_{1}} E_m\otimes I\ket{\varphi_{4}}=0$. We obtain $a_{1,2}=a_{2,1}=0$. Since $E_m^{\dagger}=E_m$, the off-diagonal elements of $E_m$ are all zero. For diagonal elements of $E_m$, we choose  $\{(1,0),(2,1)\}$. Since $\bra{\varphi_{1}} E_m\otimes I\ket{\varphi_{2}}=0$, it implies $a_{1,1}=a_{2,2}$.  We can also obtain $a_{0,0}=a_{1,1}$ by  $\{(0,5),(1,6)\}$. Thus the diagonal elements of $E_m$ are all equal. It means that $E_m$ is proportional to the identity matrix, and hence Alice cannot go first.

Let Bob and Charlie go first and start the orthogonality-preserving POVM,
\begin{equation}
E_m'=
\begin{pmatrix}
b_{0,0} & b_{0,1}  &b_{0,2} &b_{0,3} &b_{0,4} &b_{0,5} &b_{0,6} &b_{0,7} &b_{0,8}\\
b_{1,0} & b_{1,1}  &b_{1,2} &b_{1,3} &b_{1,4} &b_{1,5} &b_{1,6} &b_{1,7} &b_{1,8}\\
b_{2,0} & b_{2,1}  &b_{2,2} &b_{2,3} &b_{2,4} &b_{2,5} &b_{2,6} &b_{2,7} &b_{2,8}\\
b_{3,0} & b_{3,1}  &b_{3,2} &b_{3,3} &b_{3,4} &b_{3,5} &b_{3,6} &b_{3,7} &b_{3,8}\\
b_{4,0} & b_{4,1}  &b_{4,2} &b_{4,3} &b_{4,4} &b_{4,5} &b_{4,6} &b_{4,7} &b_{4,8}\\
b_{5,0} & b_{5,1}  &b_{5,2} &b_{5,3} &b_{5,4} &b_{5,5} &b_{5,6} &b_{5,7} &b_{5,8}\\
b_{6,0} & b_{6,1}  &b_{6,2} &b_{6,3} &b_{6,4} &b_{6,5} &b_{6,6} &b_{6,7} &b_{6,8}\\
b_{7,0} & b_{7,1}  &b_{7,2} &b_{7,3} &b_{7,4} &b_{7,5} &b_{7,6} &b_{7,7} &b_{7,8}\\
b_{8,0} & b_{8,1}  &b_{8,2} &b_{8,3} &b_{8,4} &b_{8,5} &b_{8,6} &b_{8,7} &b_{8,8}\\
\end{pmatrix},
\end{equation}
where $E_m'=M_2^\dagger M_2$. Then the states of $\{I\otimes M_2\ket{\varphi_{k}}\}_{k=1}^{24}$ are mutually orthogonal. In order to show that the off-diagonal elements of $E_m'$ are all zero, we need to choose the cells with same row index in Fig.~\ref{Figure:tite39}. If we choose the cell sets $\{(2,0),(1,1)\}$ and $\{(2,1),(1,0)\}$, then $\bra{\varphi_{3}} I\otimes E_m'\ket{\varphi_{1}}=\bra{\varphi_{3}} I\otimes E_m'\ket{\varphi_{2}}=0$. It implies $b_{0,1}=0$,  where $0$ and $1$ are the column indices of cells $(2,0)$ and $(2,1)$  respectively. In the same way, we obtain $b_{0,j}=0$ by cells $(2,0)$ and $(2,j)$ for $1\leq j\leq 7$. We also obtain $b_{0,8}=0$ by cells $(1,0)$ and $(1,8)$. In the same way, we obtain $b_{\ell,j}=0$  for $0\leq \ell\leq 3$ and $\ell+1\leq j\leq 8$. Since Fig.~\ref{Figure:tite39} is centrosymmetric,  we have $b_{8-\ell,8-j}=b_{\ell,j}=0$  for $0\leq \ell\leq 3$ and $\ell+1\leq j\leq 8$. For example, $b_{5,4}=b_{3,4}=0$, $b_{6,4}=b_{2,4}=0$, $b_{7,4}=b_{1,4}=0$, $b_{8,4}=b_{0,4}=0$. Thus the  off-diagonal elements of $E_m'$ are all zero.  For diagonal elements of $E_m'$, if we choose  $\{(1,0),(2,1)\}$, then $\bra{\varphi_{1}} I\otimes E_m'\ket{\varphi_{2}}=0$. It implies $b_{0,0}=b_{1,1}$. In the same way, we obtain $b_{2,2}=b_{3,3}$ by $\{(1,2),(2,3)\}$, $b_{1,1}=b_{3,3}$ by $\{(0,1),(0,3)\}$, and $b_{2,2}=b_{4,4}$ by $\{(0,2),(0,4)\}$. It implies $b_{0,0}=b_{1,1}=b_{2,2}=b_{3,3}=b_{4,4}$.  Since Fig.~\ref{Figure:tite39} is centrosymmetric, we can also obtain $b_{4,4}=b_{5,5}=b_{6,6}=b_{7,7}=b_{8,8}$. Thus the diagonal elements of $E_m'$ are all equal. It means that $E_m'$ is proportional to the identity matrix, and hence Bob and Charlie cannot go first.
		
We obtain that $\{\ket{\varphi_k}\}_{k=1}^{24}$ given by Eq.~(\ref{eq:tile39}) is locally irreducible. It means that $\{\ket{\psi_k}\}_{k=1}^{24}$ given by Eq.~(\ref{eq:tile333})  is locally irreducible across $A|BC$ bipartition. Further, $\{\ket{\psi_k}\}_{k=1}^{24}$ is also locally irreducible across the bipartitions $B|AC$ and $C|AB$, due to the  symmetry of  Fig.~\ref{Figure:tite333}. Thus, the orthogonal entangled set $\{\ket{\psi_k}\}_{k=1}^{24}$ given by Eq.~(\ref{eq:tile333}) is strongly nonlocal.
\end{proof}
\vspace{0.4cm}
	
Next, we give the construction of an SNOES and an SNOEB in $ 4\otimes 4\otimes 4$.

 \begin{figure}[b]
 	\centering
 	\includegraphics[scale=0.32]{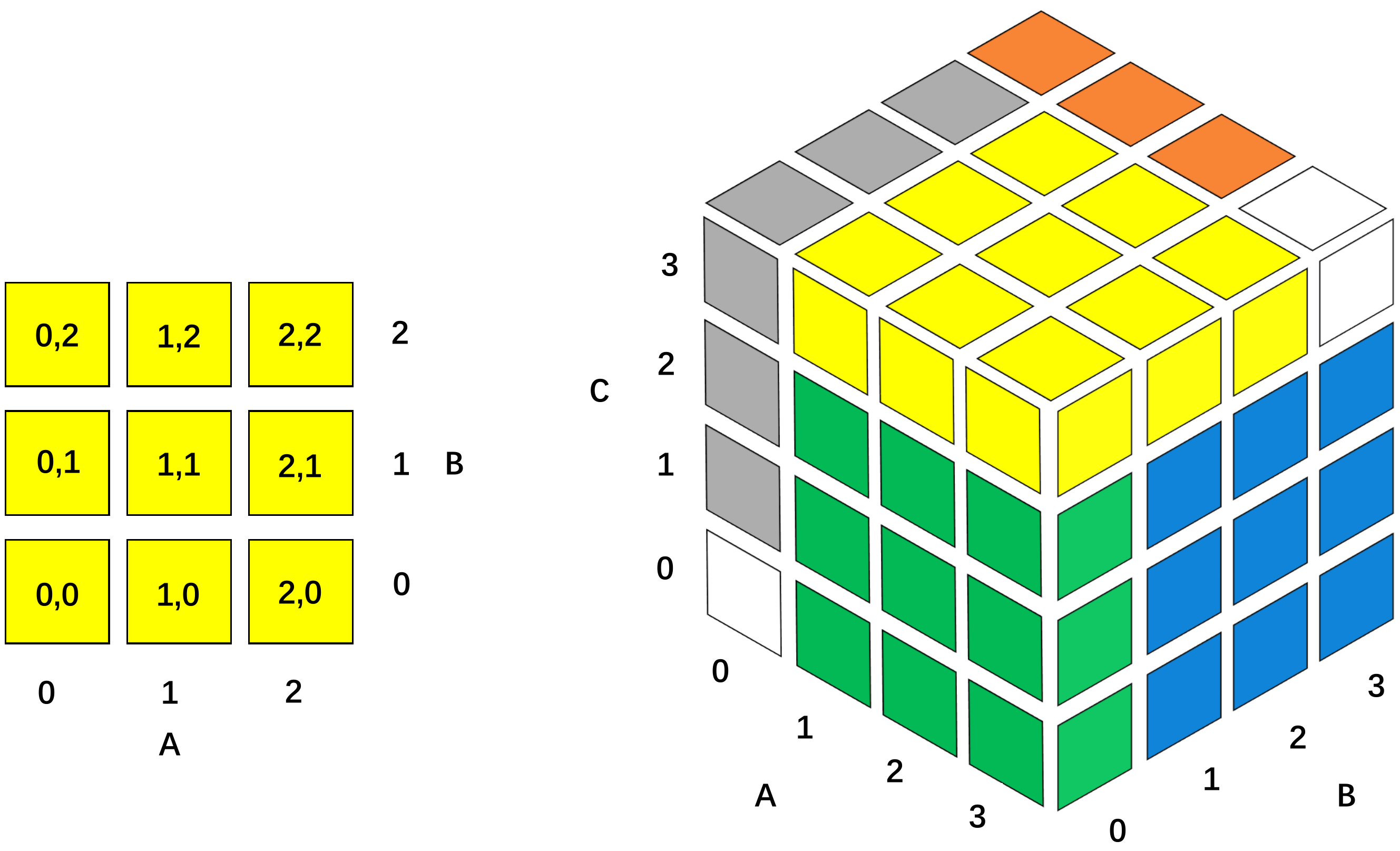}
 	\caption{A $3\times 3$ grid, and a $4\times 4\times 4$
 		Rubik's cube with six $3\times 3\times 1$ subcubes, $\{1,2,3\}\times \{0\}\times\{0,1,2\}$, $\{1,2,3\}\times \{0,1,2\}\times\{3\}$, $\{3\}\times \{1, 2, 3\}\times\{0,1,2\}$,  $\{0,1,2\}\times \{1,2,3\}\times\{0\}$,   $\{0,1,2\}\times \{3\}\times\{1,2,3\}$,  $\{0\}\times \{0,1,2\}\times\{1,2,3\}$, two  $1\times 1\times 1$ subcubes,   $\{0\}\times\{0\}\times\{0\}$,  $\{3\}\times\{3\}\times\{3\}$, and one $2\times 2\times 2$ subcube,  $\{1,2\}\times\{1,2\}\times\{1,2\}$.}  \label{Figure:tite444}
 \end{figure}

\subsection{An SNOES and an SNOEB in $ 4\otimes 4\otimes 4$}

From the left figure of Fig.~\ref{Figure:tite444}, we can obtain an $\ghz$ basis in $ 3\otimes 3$,
\begin{equation}
\begin{aligned}
\{&\{\ket{0,0}+w_3^s\ket{1,1}+w_3^{2s}\ket{2,2}\}_{s=0,1,2},\\ 
& \{\ket{0,1}+w_3^s\ket{1,2}+w_3^{2s}\ket{2,0}\}_{s=0,1,2},\\
 &\{\ket{0,1}+w_3^s\ket{1,2}+w_3^{2s}\ket{2,0}\}_{s=0,1,2}\}.
\end{aligned}
\end{equation}
There are six $3\times 3\times 1$ subcubes in the right figure of Fig.~\ref{Figure:tite444}, $\{1,2,3\}\times \{0\}\times\{0,1,2\}$, $\{1,2,3\}\times \{0,1,2\}\times\{3\}$, $\{3\}\times \{1, 2, 3\}\times\{0,1,2\}$,  $\{0,1,2\}\times \{1,2,3\}\times\{0\}$,   $\{0,1,2\}\times \{3\}\times\{1,2,3\}$,  $\{0\}\times \{0,1,2\}\times\{1,2,3\}$. By the similar construction as Eq.~(\ref{eq:tile333}), we can obtain an orthogonal entangled set in $ 4\otimes 4\otimes 4$ from these six $3\times 3\times 1$ subcubes,

\begin{equation}\label{eq:tile444}
\begin{aligned}
\ket{\psi_{1,2,3}}&=\ket{1,0,0}+w_3^s\ket{2,0,1}+w_3^{2s}\ket{3,0,2},   \\  \ket{\psi_{4,5,6}}&=\ket{1,0,1}+w_3^s\ket{2,0,2}+w_3^{2s}\ket{3,0,0},   \\
\ket{\psi_{7,8,9}}&=\ket{1,0,2}+w_3^s\ket{2,0,0}+w_3^{2s}\ket{3,0,1},  \\
\ket{\psi_{10,11,12}}&=\ket{1,0,3}+w_3^s\ket{2,1,3}+w_3^{2s}\ket{3,2,3}, \\ 
\ket{\psi_{13,14,15}}&=\ket{1,1,3}+w_3^s\ket{2,2,3}+w_3^{2s}\ket{3,0,3},  \\
\ket{\psi_{16,17,18}}&=\ket{1,2,3}+w_3^s\ket{2,0,3}+w_3^{2s}\ket{3,1,3},  \\
\ket{\psi_{19,20,21}}&=\ket{3,1,0}+w_3^s\ket{3,2,1}+w_3^{2s}\ket{3,3,2},  \\
\ket{\psi_{22,23,24}}&=\ket{3,1,1}+w_3^s\ket{3,2,2}+w_3^{2s}\ket{3,3,0},   \\
\ket{\psi_{25,26,27}}&=\ket{3,1,2}+w_3^s\ket{3,2,0}+w_3^{2s}\ket{3,3,1},  \\
\ket{\psi_{28,29,30}}&=\ket{0,1,0}+w_3^s\ket{1,2,0}+w_3^{2s}\ket{2,3,0},  \\
\ket{\psi_{31,32,33}}&=\ket{0,2,0}+w_3^s\ket{1,3,0}+w_3^{2s}\ket{2,1,0},  \\
\ket{\psi_{34,35,36}}&=\ket{0,3,0}+w_3^s\ket{1,1,0}+w_3^{2s}\ket{2,2,0}, \\
\ket{\psi_{37,38,39}}&=\ket{0,3,1}+w_3^s\ket{1,3,2}+w_3^{2s}\ket{2,3,3}, \\
\ket{\psi_{40,41,42}}&=\ket{0,3,2}+w_3^s\ket{1,3,3}+w_3^{2s}\ket{2,3,1},  \\
\ket{\psi_{43,44,45}}&=\ket{0,3,3}+w_3^s\ket{1,3,1}+w_3^{2s}\ket{2,3,2}, \\
\ket{\psi_{46,47,48}}&=\ket{0,0,1}+w_3^s\ket{0,1,2}+w_3^{2s}\ket{0,2,3},   \\ 
\ket{\psi_{49,50,51}}&=\ket{0,0,2}+w_3^s\ket{0,1,3}+w_3^{2s}\ket{0,2,1},   \\
\ket{\psi_{52,53,54}}&=\ket{0,0,3}+w_3^s\ket{0,1,1}+w_3^{2s}\ket{0,2,2}, 
\end{aligned}
\end{equation}	
where $s=0,1,2$, $\ket{\psi_{j_1,j_2,j_3}}=\ket{k_1,k_2,k_3}+w_3^s\ket{k_4,k_5,k_6}+w_3^{2s}\ket{k_7,k_8,k_9}$ means  $\ket{\psi_{j_1}}=\ket{k_1,k_2,k_3}+\ket{k_4,k_5,k_6}+\ket{k_7,k_8,k_9}$, $\ket{\psi_{j_2}}=\ket{k_1,k_2,k_3}+w_3\ket{k_4,k_5,k_6}+w_3^2\ket{k_7,k_8,k_9}$, and $\ket{\psi_{j_3}}=\ket{k_1,k_2,k_3}+w_3^2\ket{k_4,k_5,k_6}+w_3\ket{k_7,k_8,k_9}$. We can also extend the orthogonal entangled set  $\{\ket{\psi_k}\}_{k=1}^{54}$ given by Eq.~(\ref{eq:tile444}) to an orthogonal entangled basis in $ 4\otimes 4\otimes 4$. Since there are two  $1\times 1\times 1$ subcubes,   $\{0\}\times\{0\}\times\{0\}$,  $\{3\}\times\{3\}\times\{3\}$, and one $2\times 2\times 2$ subcube,  $\{1,2\}\times\{1,2\}\times\{1,2\}$ left from the right figure of Fig.~\ref{Figure:tite444}, we can choose the following  genuinely entangled orthogonal states,
\begin{equation}\label{eq:tenghz}
\begin{aligned}
\ket{\psi_{55,56,57,58}}=&\{\ket{0,0,0}+w_4^k\ket{1,1,1}+w_4^{2k}\ket{2,2,2}\\
&+w_4^{3k}\ket{3,3,3}\}_{k=0}^{3},\\
\ket{\psi_{59,60}}=&\ket{1,2,2}\pm\ket{2,1,1},\\
 \ket{\psi_{61,62}}=&\ket{1,1,2}\pm\ket{2,2,1},\\
\ket{\psi_{63,64}}=&\ket{1,2,1}\pm\ket{2,1,2}.
\end{aligned}
\end{equation}
Then  $\{\psi_k\}_{k=1}^{64}$ given by Eqs.~(\ref{eq:tile444}) and (\ref{eq:tenghz}) forms an orthogonal entangled basis in $ 4\otimes 4\otimes 4$. In the following, we show that  $\{\ket{\psi_k}\}_{k=1}^{54}$ and $\{\psi_k\}_{k=1}^{64}$ are both strongly nonlocal.
	
\begin{lemma}\label{lem:strong444}
In $ 4\otimes 4\otimes 4$, the orthogonal entangled set $\{\ket{\psi_k}\}_{k=1}^{54}$ and the orthogonal entangled basis $\{\psi_k\}_{k=1}^{64}$ given by Eqs.~(\ref{eq:tile444}) and (\ref{eq:tenghz})  are both strongly nonlocal.
\end{lemma}

The proof of Lemma~\ref{lem:strong444} is given in Appendix~\ref{appendix:lem444}. In next section, we consider SNOESs and SNOEBs in $ d\otimes d\otimes d$ for $d\geq 3$.
\begin{figure}[b]
	\centering
	\includegraphics[scale=0.4]{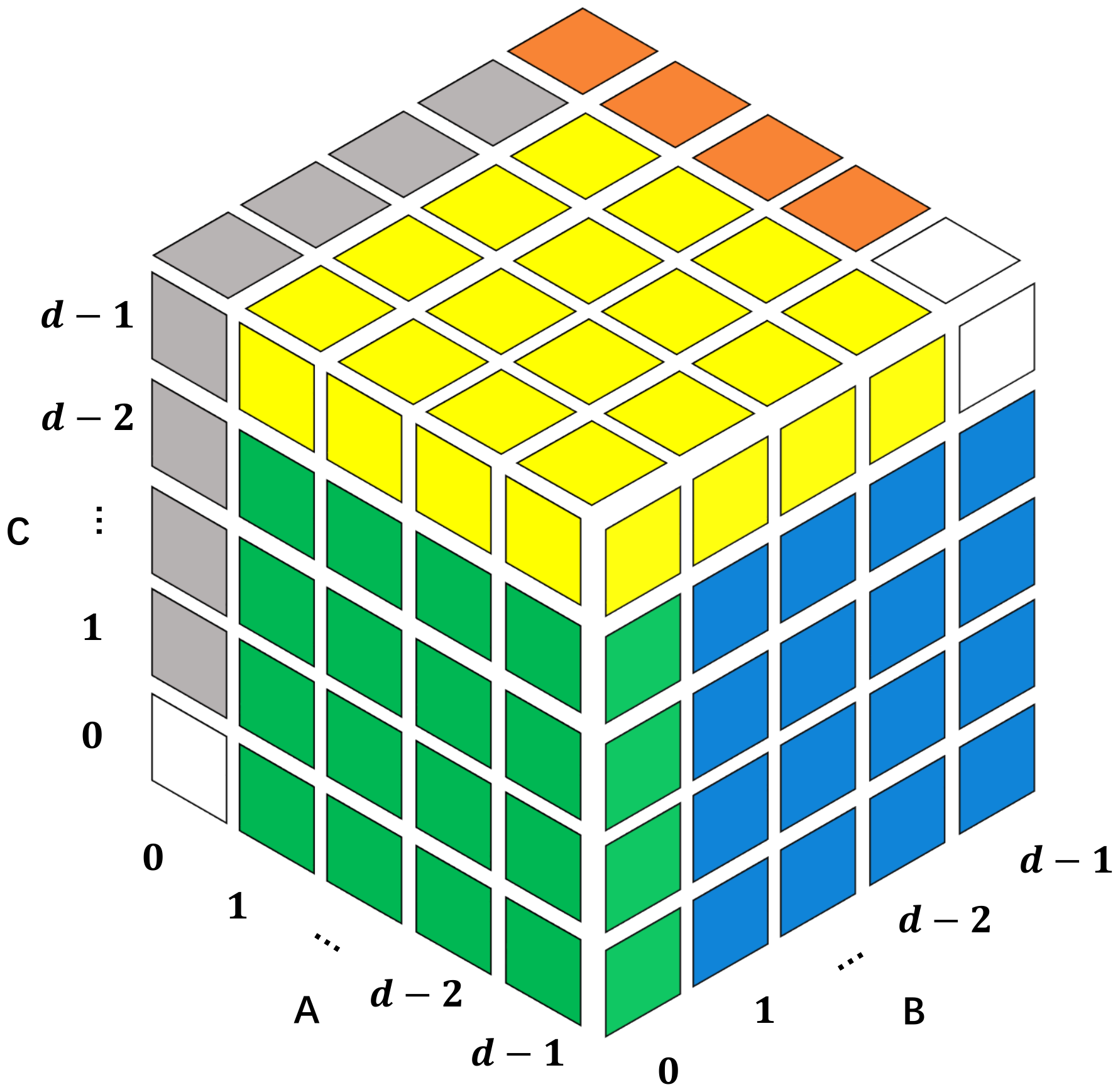}
	\caption{$d\times d\times d$ Rubik's cube, $d\geq 3$. Our construction is like peeling onions. We start from the most outside layer, for which the cells are partitioned as above. After we peel this layer, we get a $(d-2)\times(d-2)\times(d-2)$ Rubik's cube, and we can similarly partition the cells as above. Continue this procedure to the core, until we get a $3\times 3\times 3$ Rubik's cube when $d$ is odd, or a $4\times4\times4$ Rubik's cube when $d$ is even. Then we apply Fig.~\ref{Figure:tite333} and Fig.~\ref{Figure:tite444} to these two cores to get a complete partition of the $d\times d\times d$ Rubik's cube.}\label{Figure:555666}
\end{figure}	

\subsection{SNOESs and SNOEBs in $ d\otimes d\otimes d$ for $d\geq 3$}

 We give a general construction of SNOESs and SNOEBs in $ d\otimes d\otimes d$ for $d\geq 3$. Our construction is like peeling onions. We start from the most outside layer, for which the cells are partitioned as in Fig.~\ref{Figure:555666}. After we peel this layer, we get a $(d-2)\times(d-2)\times(d-2)$ Rubik's cube, and we can similarly partition the cells as in Fig.~\ref{Figure:555666}. Continue this procedure to the core, until we get a $3\times 3\times 3$ Rubik's cube when $d$ is odd, or a $4\times4\times4$ Rubik's cube when $d$ is even. Then we apply Fig.~\ref{Figure:tite333} and Fig.~\ref{Figure:tite444} to these two cores to get a complete partition of the $d\times d\times d$ Rubik's cube in Fig.~\ref{Figure:555666}.  In order to apply Lemmas~\ref{lem:strong} and \ref{lem:strong444} to the cores, we need to match the coordinates. Define a bijection: $\ket{0}\rightarrow\ket{\frac{d-3}{2}}$,  $\ket{1}\rightarrow\ket{\frac{d-1}{2}}$, $\ket{2}\rightarrow\ket{\frac{d+1}{2}}$ when $d$ is odd, and a bijection: $\ket{0}\rightarrow\ket{\frac{d-4}{2}}$,  $\ket{1}\rightarrow\ket{\frac{d-2}{2}}$,  $\ket{2}\rightarrow\ket{\frac{d}{2}}$,  $\ket{3}\rightarrow\ket{\frac{d+2}{2}}$  when $d$ is even. Then Eq.~(\ref{eq:tile333}) and  Eq.~(\ref{eq:tile444}) are mapped to orthogonal entangled sets $\{\ket{\psi_k'}\}_{k=1}^{24}$  and $\{\ket{\psi_k'}\}_{k=1}^{54}$ respectively.


Now, we have the following orthogonal entangled set in $ d\otimes d\otimes d$ based on Fig.~\ref{Figure:555666}:
\begin{equation}\label{eq:tileddd}
\cB^{(d)}=\left\{
\begin{aligned}
\{&\ket{\psi_{(k_1,t_1)}}=\sum_{j=0}^{d-2}w_{d-1}^{jk_1}\ket{j+1}\ket{0}\ket{j\oplus_{d-1}t_1},\\
&{0\leq k_1,t_1\leq d-2},\\
&\ket{\psi_{(k_2,t_2)}}=\sum_{j=0}^{d-2}w_{d-1}^{jk_2}\ket{j+1}\ket{j\oplus_{d-1}t_2}\ket{d-1},\\
&{0\leq k_2,t_2\leq d-2}, \\
&\ket{\psi_{(k_3,t_3)}}=\sum_{j=0}^{d-2}w_{d-1}^{jk_3}\ket{d-1}\ket{j+1}\ket{j\oplus_{d-1}t_3},\\
&{0\leq k_3,t_3\leq d-2}, \\
&\ket{\psi_{(k_4,t_4)}}=\sum_{j=0}^{d-2}w_{d-1}^{jk_4}\ket{j}\ket{(j\oplus_{d-1}t_4)+1}\ket{0},\\
&{0\leq k_4,t_4\leq d-2}, \\
&\ket{\psi_{(k_5,t_5)}}=\sum_{j=0}^{d-2}w_{d-1}^{jk_5}\ket{j}\ket{d-1}\ket{(j\oplus_{d-1}t_5)+1},\\
&{0\leq k_5,t_5\leq d-2},\\
&\ket{\psi_{(k_6,t_6)}}=\sum_{j=0}^{d-2}w_{d-1}^{jk_6}\ket{0}\ket{j}\ket{(j\oplus_{d-1}t_6)+1},\\
&{0\leq k_6,t_6\leq d-2}, \\
&\ldots\\
& \{\ket{\psi_k'}\}_{k=1}^{24} \ \text{when $d$ is odd}, \\
& (\{\ket{\psi_k'}\}_{k=1}^{54}  \ \text{when $d$ is even})\},
\end{aligned}\right.
\end{equation}
where $j\oplus_{d-1} t=j+t\mod(d-1)$. When $d$ is odd, the number of the entangled states in $\cB^{(d)}$ is $6[(d-1)^2+(d-3)^2+\cdots+2^2]=d^3-d$. We can add $d$ genuinely entangled states into $\cB^{(d)}$ to form an orthogonal entangled basis: 
\begin{equation}\label{eq:addodd}
\cB_1=\left\{\sum_{j=0}^{d-1}w_d^{jk}\ket{j,j,j}\right\}_{k=0}^{d-1}.
\end{equation}
When $d$ is even, the number of the entangled states in $\cB^{(d)}$ is $6[(d-1)^2+(d-3)^2+\cdots+3^2]=d^3-d-6$.  We can also add $d+6$ genuinely entangled states into $\cB^{(d)}$ to form an orthogonal entangled basis:
\begin{equation}\label{eq:addeven}
\begin{split}
\cB_2&=\left\{\sum_{j=0}^{d-1}w_d^{jk}\ket{j,j,j}\right\}_{k=0}^{d-1}\\
&\bigcup\{\ket{m-1,m,m}\pm\ket{m,m-1,m-1}\}\\
&\bigcup\{\ket{m-1,m-1,m}\pm\ket{m,m,m-1}\}\\
&\bigcup\{\ket{m-1,m,m-1}\pm\ket{m,m-1,m}\},
\end{split}
\end{equation}
where $m=\frac{d}{2}$.
Now, we show that $\cB^{(d)}$, $\cB^{(d)}\cup \cB_1$ and $\cB^{(d)}\cup \cB_2$ are strongly nonlocal.
	
\begin{theorem}\label{thm:generalddd}
In $ d\otimes d\otimes d$, $d\geq 3$,
\begin{enumerate}[(i)]
\item when $d$ is odd,  $\cB^{(d)}$ with $|\cB^{(d)}|=d^3-d$ given by Eq.~(\ref{eq:tileddd}) is an SNOES,  and  $\cB^{(d)}\cup \cB_1$ given by Eq.~(\ref{eq:tileddd}) and Eq.~(\ref{eq:addodd}) is an SNOEB;
\item when $d$ even,  $\cB^{(d)}$ with $|\cB^{(d)}|=d^3-d-6$ given by Eq.~(\ref{eq:tileddd}) is an SNOES,  and  $\cB^{(d)}\cup \cB_2$ given by Eq.~(\ref{eq:tileddd}) and Eq.~(\ref{eq:addeven}) is an SNOEB.
\end{enumerate}
\end{theorem}

The proof of Theorem~\ref{thm:generalddd} is given in Appendix~\ref{appendix:lemddd}. By Theorem~\ref{thm:generalddd}, we can find SNOEBs in $d\otimes d\otimes d$ for $d\geq 3$, which answers an open question in Ref.~\cite{Halder2019Strong}. From the proof of Theorem~\ref{thm:generalddd}, we know that if there exists an SNOES of size $s$ in $d\otimes d\otimes d$, then there exists an SNOES of size $s+6(d+1)^2$ in $(d+2)\otimes (d+2)\otimes (d+2)$. A similar structure of Fig.~\ref{Figure:555666} also appears in Ref.~\cite{Agrawal2019Genuinely}, where the authors used it to construct unextendible product bases (UPBs).

\section{Entanglement-assisted discrimination}
\label{sec:discri-entang}
	
In this section, we consider local discrimination of SNOESs using entanglement as a resource.  Since SNOESs cannot be locally distinguished in every bipartition, a perfect local discrimination of this set would require a resource state that must be entangled in all bipartitions. Assume  $\cB$ is an SNOES in $ d\otimes d\otimes d$. Let Alice and Bob share a  maximally entangled state (MES) $\ket{\phi(d)}=\sum_{k=0}^{d-1}\ket{k,k}$, Alice and Charlie also share $\ket{\phi(d)}$. Using the  MES $\ket{\phi(d)}$, Bob (Charlie) can  teleport his subsystem to Alice  \cite{Bennett1993Teleporting,Sumit2019Genuinely}. Then Alice  can perfectly discriminate $\cB$ by performing a suitable measurement. In this teleportation-based protocol, it consumes $2\log_2 d$ ebits entanglement resource.  A protocol consuming less entanglement than the teleportation-based protocol is desirable, since entanglement is a costly resource under the operational paradigm of LOCC. We give two entanglement-assisted discrimination protocols for the SNOES in $ 3\otimes 3\otimes 3$ in Proposition~\ref{pro:distinguish39} and Proposition~\ref{pro:distinguish333}, Each  protocol consumes less entanglement resource than the teleportation-based protocol averagely.
	
First, we give an entanglement-assisted discrimination protocol for Bell basis in $ 2\otimes 2$ by using a two-qubit MES. It is different from the teleportation-based protocol.
\begin{example}\label{example:belldisting}
The Bell basis in $ 2\otimes 2$ can be locally distinguished by using a two-qubit MES.	
The initial states are
\begin{equation}
\begin{aligned}
\ket{\psi_{1,2}}&=(\ket{0,0}\pm\ket{1,1})_{A,B}(\ket{0,0}+\ket{1,1})_{a,b},  \\
\ket{\psi_{3,4}}&=(\ket{0,1}\pm\ket{1,0})_{A,B}(\ket{0,0}+\ket{1,1})_{a,b},
\end{aligned}
\end{equation}
where $a$ and $b$ are the ancillary systems of Alice and Bob, respectively.
Denote  $P[\ket{i}_\spadesuit]:=\ketbra{i}{i}_\spadesuit$,  $P[\ket{i}_\spadesuit;\ket{j}_\clubsuit]:=\ketbra{i}{i}_\spadesuit\otimes\ketbra{j}{j}_\clubsuit$, and $P[(\ket{k},\ket{\ell})_\spadesuit;(\ket{m},\ket{n})_\clubsuit]:=(\ketbra{k}{k}+\ketbra{\ell}{\ell})_\spadesuit\otimes (\ketbra{m}{m}+\ketbra{n}{n})_\clubsuit$.  Now the discrimination protocol proceeds as follows.
		
\textit{Step 1}.  Alice performs the measurement $\{N_1:=P[\ket{0}_A;\ket{0}_a]+P[\ket{1}_A;\ket{1}_a], \overline{N_1}=I-N_1\}$.  If $N_1$ clicks, the resulting postmeasurement states are
\begin{equation}
\begin{aligned}
\ket{\psi_{1,2}}&\rightarrow\ket{0,0}_{A,B}\ket{0,0}_{a,b}\pm\ket{1,1}_{A,B}\ket{1,1}_{a,b},  \\
\ket{\psi_{3,4}}&\rightarrow\ket{0,1}_{A,B}\ket{0,0}_{a,b}\pm\ket{1,0})_{A,B}\ket{1,1}_{a,b}.
\end{aligned}
\end{equation}
		
\textit{Step 2}.  Bob performs the measurement $\{N_2:=P[\ket{0}_B;\ket{0}_b]+P[\ket{1}_B;\ket{1}_b], \overline{N_2}=I-N_2\}$. If $N_2$ clicks, it remains $\ket{\psi_{1,2}}$, which can be locally distinguished \cite{walgate2000local}.  Otherwise, he performs $\overline{N_2}$, and it remains the locally distinguishable set $\ket{\psi_{3,4}}$. If $\overline{N_1}$ clicks in the step 1, we can obtain a  similar protocol.
	\end{example}
	
For a tripartite system,    the configuration of entanglement resources can be described by $\{(p,\ket{\phi(d_1)})_{A,B}, (q,\ket{\phi(d_2)})_{A,C},(r,\ket{\phi(d_3)})_{B,C}\}$ \cite{Sumit2019Genuinely}, where $(p,\ket{\phi(d_1)})_{A,B}$ means that an amount $p$ of the MES $\ket{\phi(d_1)}=\sum_{k=0}^{d_1-1}\ket{k,k}$  is consumed between Alice and Bob averagely, and similarly for $(q,\ket{\phi(2)})_{A,C} $ and $(r,\ket{\phi(d_3)})_{B,C}$. Next, we give a protocol for the SNOES $\{\ket{\psi_k}\}_{k=1}^{24}$ given by Eq.~(\ref{eq:tile333}) in $3\otimes 3\otimes 3$. 
\begin{proposition}\label{pro:distinguish39}
The SNOES  $\{\ket{\psi_k}\}_{k=1}^{24}$ given by Eq.~(\ref{eq:tile333}) can be locally distinguished by using  $\{(\frac{4}{3},\ket{\phi(2)})_{A,B}, (0,\ket{\phi(2)})_{A,C},(1,\ket{\phi(3)})_{B,C}\}$, where $(\frac{4}{3},\ket{\phi(2)})_{A,B}$ means that two $\ket{\phi(2)}$ are distributed between Alice and Bob, and $\frac{4}{3}$ $\ket{\phi(2)}$ are actually consumed.
\end{proposition}
\begin{proof}
First, Charlie teleports his subsystem to Bob by using  the entanglement resource $\ket{\phi(3)}$. Then $\{\ket{\psi_k}\}_{k=1}^{24}$ given by Eq.~(\ref{eq:tile333}) is transformed into  $\{\ket{\varphi_k}\}_{k=1}^{24}$ given by Eq.~(\ref{eq:tile39}). We use the subindex $\wi{B}$ for this union of Bob and Charlie. The  two $\ket{\phi(2)}$ are distributed between Alice and Bob.  The initial state is
\begin{equation}
\ket{\varphi}_{{A,\wi{B}}}\otimes\ket{\phi(2)}_{a,b}\otimes\ket{\phi(2)}_{a_1,b_1}.
\end{equation}
where $a$ and $a_1$ are the ancillary systems of Alice, $b$ and $b_1$ are the ancillary systems of Bob. Now the discrimination protocol proceeds as follows.
		
\textit{Step 1}.   Alice performs the measurement $\{K_1:=P[\ket{0}_A;\ket{0}_a]+P[(\ket{1},\ket{2})_A;\ket{1}_a], \overline{K_1}:=I-K_1\}$.  If $K_1$ clicks, the resulting postmeasurement states are
\begin{equation*}
\begin{aligned}
\ket{\psi_{1,2}}\rightarrow&(\ket{1,0}\pm\ket{2,1})_{{A,\wi{B}}}\ket{1,1}_{a,b}\ket{\phi(2)}_{a_1,b_1}, \\ \ket{\psi_{3,4}}\rightarrow&(\ket{1,1}\pm\ket{2,0})_{{A,\wi{B}}}\ket{1,1}_{a,b}\ket{\phi(2)}_{a_1,b_1} ,\\
\ket{\psi_{5,6}}\rightarrow&(\ket{1,2}\pm\ket{2,3})_{{A,\wi{B}}}\ket{1,1}_{a,b}\ket{\phi(2)}_{a_1,b_1}, \\ \ket{\psi_{7,8}}\rightarrow&(\ket{1,3}\pm\ket{2,2})_{{A,\wi{B}}}\ket{1,1}_{a,b}\ket{\phi(2)}_{a_1,b_1} ,\\
\ket{\psi_{9,10}}\rightarrow&(\ket{2,5}\pm\ket{2,7})_{{A,\wi{B}}}\ket{1,1}_{a,b}\ket{\phi(2)}_{a_1,b_1},\\ \ket{\psi_{11,12}}\rightarrow&(\ket{2,4}\pm\ket{2,6})_{{A,\wi{B}}}\ket{1,1}_{a,b}\ket{\phi(2)}_{a_1,b_1} ,\\
\ket{\psi_{13,14}}\rightarrow&(\ket{0,5}_{{A,\wi{B}}}\ket{0,0}_{a,b}\pm\ket{1,6}_{{A,\wi{B}}}\ket{1,1}_{a,b})\\
&\otimes\ket{\phi(2)}_{a_1,b_1},\\
\ket{\psi_{15,16}}\rightarrow&(\ket{0,6}_{{A,\wi{B}}}\ket{0,0}_{a,b}\pm\ket{1,5}_{{A,\wi{B}}}\ket{1,1}_{a,b})\\
&\otimes\ket{\phi(2)}_{a_1,b_1},
\end{aligned}
\end{equation*}
\begin{align}
\notag
\ket{\psi_{17,18}}\rightarrow&(\ket{0,7}_{{A,\wi{B}}}\ket{0,0}_{a,b}\pm\ket{1,8}_{{A,\wi{B}}}\ket{1,1}_{a,b})\\
\notag
&\otimes\ket{\phi(2)}_{a_1,b_1},\\
\notag
 \ket{\psi_{19,20}}\rightarrow&(\ket{0,8}_{{A,\wi{B}}}\ket{0,0}_{a,b}\pm\ket{1,7}_{{A,\wi{B}}}\ket{1,1}_{a,b})\\
 \notag
&\otimes\ket{\phi(2)}_{a_1,b_1} ,\\
\notag
\ket{\psi_{21,22}}\rightarrow&(\ket{0,1}\pm\ket{0,3})_{{A,\wi{B}}}\ket{0,0}_{a,b}\ket{\phi(2)}_{a_1,b_1},\\ \ket{\psi_{23,24}}\rightarrow&(\ket{0,2}\pm\ket{0,4})_{{A,\wi{B}}}\ket{0,0}_{a,b}\ket{\phi(2)}_{a_1,b_1}.
\end{align}

\textit{Step 2}.   Bob performs the measurement $\{K_{2,1}:=P[(\ket{1},\ket{3})_{\wi{B}};\ket{0}_b], K_{2,2}:=P[(\ket{2},\ket{4})_{\wi{B}};\ket{0}_b], K_{2,3}:=P[\ket{7}_{\wi{B}};\ket{0}_b]+P[\ket{8}_{\wi{B}};\ket{1}_b], K_{2,4}:=P[(\ket{0},\ket{1})_{\wi{B}};\ket{1}_b], K_{2,5}:=P[(\ket{2},\ket{3})_{\wi{B}};\ket{1}_b], \overline{K_2}:=I-K_{2,1}-K_{2,2}-K_{2,3}-K_{2,4}-K_{2,5}\}$. If $K_{2,1}$ clicks, it remains $\ket{\psi_{21,22}}$; if $K_{2,2}$ clicks, it remains $\ket{\psi_{23,24}}$; if $K_{2,3}$ clicks, it remains $\ket{\psi_{17,18}}$; If $K_{2,4}$ clicks, it remains $\ket{\psi_{1,2,3,4}}$. These four states can be locally distinguished by using $\ket{\phi(2)}_{a_1,b_1}$ (see Example~\ref{example:belldisting}). If $K_{2,5}$ clicks, it remains $\ket{\psi_{5,6,7,8}}$. These four states can also be locally distinguished by using $\ket{\phi(2)}_{a_1,b_1}$; if $\overline{K_2}$ clicks, it remains $\{\ket{\psi_k}_{k=9}^{16}\bigcup\ket{\psi_{19,20}}\}$.
			
\textit{Step 3}. Alice performs the measurement $\{K_{3}:=P[\ket{2}_A], \overline{K_3}:=I-K_{3}\}$. If $K_{3}$ clicks, it remains $\ket{\psi_{9,10,11,12}}$. Bob can distinguish these four product states; if $\overline{K_3}$ clicks, it remains $\{\ket{\psi_k}_{k=13}^{16}\bigcup\ket{\psi_{19,20}}\}$.
		
\textit{Step 4}. Bob performs the measurement $\{K_{4,1}:=P[\ket{8}_{\wi{B}};\ket{0}_b]+P[\ket{7}_{\wi{B}};\ket{1}_b], K_{4,2}:=P[\ket{5}_{\wi{B}};\ket{0}_b]+P[\ket{6}_{\wi{B}};\ket{1}_b], \overline{K_4}:=I-K_{4,1}-K_{4,2}\}$. If $K_{4,1}$ clicks, it remains $\ket{\psi_{19,20}}$; if $K_{4,2}$ clicks, it remains $\ket{\psi_{13,14}}$; if  $\overline{K_4}$ clicks, it remains $\ket{\psi_{15,16}}$.
		
If $\overline{K_1}$ clicks in step 1, we can obtain a  similar protocol. From the beginning to step 4,  it consumes $1$  $\ket{\phi(3)}$  between Alice and Charlie, and  $1+2\times \frac{4}{24}=\frac{4}{3}$ $\ket{\phi(2)}$  between Alice and Bob averagely.
\end{proof}
\vspace{0.5cm}
The protocol in Proposition~\ref{pro:distinguish39}  consumes $\frac{4}{3}+\log_23$ ebits  entanglement resource averagely, which is strictly less than $2\log_23$. It means that this protocol consumes less entanglement resource than the teleportation-based protocol. However, in the protocol in Proposition~\ref{pro:distinguish39}, the tripartite system becomes the bipartite system. Since it uses the teleportation-based protocol between Bob and Charlie. In the following, we give a  more efficient protocol when the three parties are separated.
	
\begin{proposition}\label{pro:distinguish333}
The SNOES  $\{\ket{\psi_k}\}_{k=1}^{24}$ given by Eq.~(\ref{eq:tile333}) can be locally distinguished by using  $\{(\frac{7}{6},\ket{\phi(2)})_{A,B}, (\frac{7}{6},\ket{\phi(2)})_{A,C},(\frac{1}{6},\ket{\phi(2)})_{B,C}\}$, where $(\frac{7}{6},\ket{\phi(2)})_{A,B}$ means that two $\ket{\phi(2)}$ are distributed between Alice and Bob, and $\frac{7}{6}$ $\ket{\phi(2)}$ are actually consumed, and similarly for  $(\frac{7}{6},\ket{\phi(2)})_{A,C}$ and $(\frac{1}{6},\ket{\phi(2)})_{B,C}$.
\end{proposition}

The proof of Proposition~\ref{pro:distinguish333} is given in Appendix~\ref{appendix:dis333}. The protocol in Proposition~\ref{pro:distinguish333}  consumes $\frac{5}{2}$ ebits entanglement resource, which is less than the protocol in Proposition~\ref{pro:distinguish39}. In Ref.~\cite{Sumit2019Genuinely}, the authors investigated entanglement-assisted discrimination of a strongly nonlocal orthogonal product bases (SNOPB) in $ 3\otimes 3\otimes 3$ (which  is from Ref.~\cite{Halder2019Strong}). They showed that this SNOPB can be locally distinguished by using $\{(1,\ket{\phi(3)})_{A,B}, (1,\ket{\phi(2)})_{A,C},(0,\ket{\phi(2)})_{B,C}\}$. It used the teleportation-based protocol between Alice and Bob. This protocol  consumes  $1+\log_23$ ebits entanglement resource, which is less than the protocol in Proposition~\ref{pro:distinguish39}. Moreover, they also give a protocol when the three parties are separated. They showed that this SNOPB can be locally distinguished by using $\{(1,\ket{\phi(3)})_{A,B}, (1,\ket{\phi(2)})_{A,C},(\frac{8}{27},\ket{\phi(2)})_{B,C}\}$. This protocol also consumes less entanglement resource than the protocol in Proposition~\ref{pro:distinguish333}. We find that using the configuration of their entanglement resource, we cannot obtain a perfectly discrimination protocol for our SNOES in $ 3\otimes 3\otimes 3$. Thus, entanglement-assisted discrimination of SNOES may consume more entanglement resource than that  entanglement-assisted discrimination of SNOPB in $ 3\otimes 3\otimes 3$.  It means that entanglement can increase the difficulty to locally distinguish orthogonal states.

\section{conclusion}
\label{sec:con}
We have constructed an SNOES of size $d^3-d$  in $ d\otimes d \otimes d$ when $d\geq 3$ is odd,  and an SNOES  of size $d^3-d-6$  in $ d\otimes d \otimes d$ when $d\geq 3$ is even. We have extended these SNOESs to SNOEBs, and it answers an open question in Ref.~\cite{Halder2019Strong}. We have also given two entanglement-assisted discrimination protocols for the SNOES in $ 3\otimes 3\otimes 3$. Each protocol consumes less entanglement resource than the teleportation-based protocol averagely. Our results show the phenomenon of strong quantum nonlocality with entanglement. There are some interesting problems left.  We don't know whether three qubit SNOESs exist. Another problem is how to  generalize the construction in $( d)^{\otimes n}$ for any $d\geq 2$ and $n\geq 4$.
		
\section*{Acknowledgments}
\label{sec:ack}	
FS and XZ were supported by NSFC under Grant No. 11771419,  the Fundamental Research Funds for the Central Universities, and Anhui Initiative in Quantum Information Technologies under Grant No. AHY150200. LC and MH were supported by the  NNSF of China (Grant No. 11871089), and the Fundamental Research Funds for the Central Universities (Grant No. ZG216S2005).

\appendix
\section{A lemma of linear algebra}\label{appendix:linear}
The following  lemma of linear algebra is frequently used in Appendix~\ref{appendix:lem444} and Appendix~\ref{appendix:lemddd}.
\begin{lemma}\label{lem:varder}
	Let $w_n=e^{\frac{2\pi i}{n}}$. If
	\begin{equation}
	\left\{\begin{split}
	&x_{1}+x_{2}+\ldots+x_{n}=0,\\
	&x_{1}+w_nx_{2}+\ldots+w_n^{(n-1)}x_{n}=0,\\
	&x_{1}+w_n^2x_{2}+\ldots+w_n^{2(n-1)}x_{n}=0,\\
	&\ldots\\
	&x_{1}+w_n^{n-1}x_{2}+\ldots+w_n^{(n-1)(n-1)}x_{n}=0,\\
	\end{split}
	\right.
	\end{equation}
	then $x_1=x_2=\ldots=x_n=0.$	
	If
	\begin{equation}
	\left\{\begin{split}
	&x_{1}+w_nx_{2}+\ldots+w_n^{(n-1)}x_{n}=0,\\
	&x_{1}+w_n^2x_{2}+\ldots+w_n^{2(n-1)}x_{n}=0,\\
	&\ldots\\
	&x_{1}+w_n^{n-1}x_{2}+\ldots+w_n^{(n-1)(n-1)}x_{n}=0,\\
	\end{split}
	\right.
	\end{equation}
	then $x_1=x_2=\ldots=x_n$.	
\end{lemma}
\section{The proof of Lemma~\ref{lem:strong444}}\label{appendix:lem444}	
\begin{figure*}[t]
	\centering
	\includegraphics[scale=0.9]{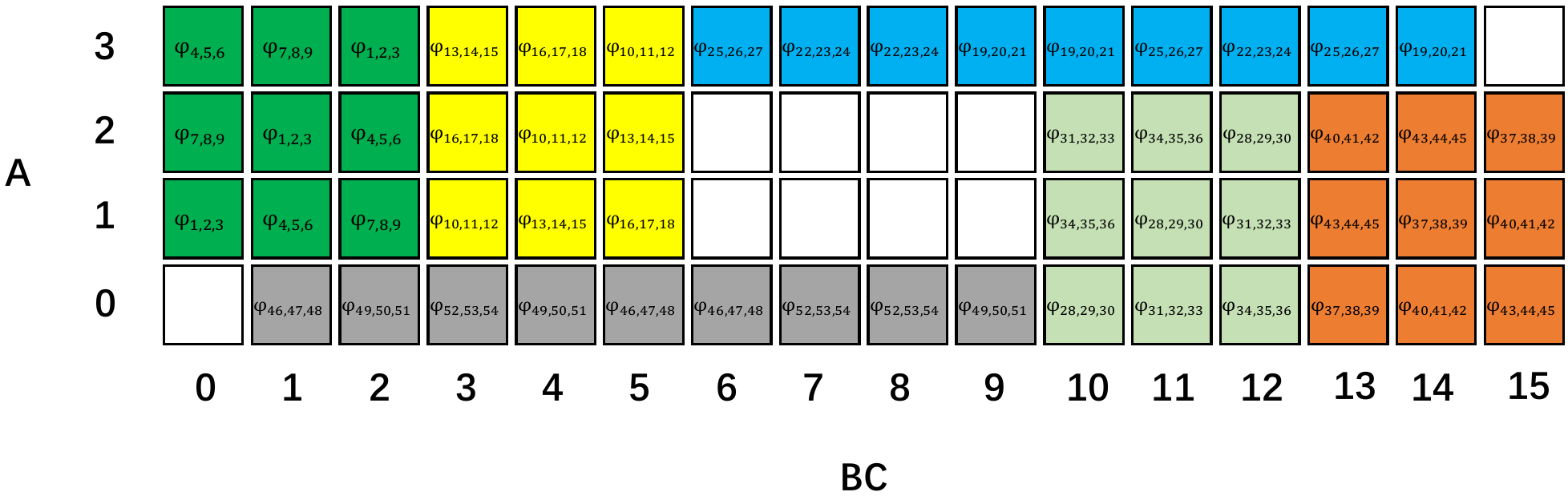}
	\caption{The corresponding $4\times 16$ grid of Eq.~(\ref{eq:tile416}).  For example, $\ket{\varphi_{1,2,3}}$ correspond to the cell set $\{(1,0),(2,1),(3,2)\}$.  }  \label{Figure:tite416}
\end{figure*}

\begin{proof}
	We only need to show that  $\{\ket{\psi_k}\}_{k=1}^{54}$ given by Eq.~(\ref{eq:tile444}) is strongly nonlocal. Since Fig.~\ref{Figure:tite444} is symmetric, we only need to consider $A|BC$ bipartition.  Define a bijection from the basis  $\{\ket{p,q}\}_{p,q=0}^{3}$ in $ 4\otimes 4$ to the  basis in $\bbC^{16}$ as follows: $\ket{0,0}\rightarrow\ket{0}$, $\ket{0,1}\rightarrow\ket{1}$, $\ket{0,2}\rightarrow\ket{2}$, $\ket{0,3}\rightarrow\ket{3}$, $\ket{1,0}\rightarrow\ket{10}$, $\ket{1,1}\rightarrow\ket{7}$,  $\ket{1,2}\rightarrow\ket{6}$, $\ket{1,3}\rightarrow\ket{4}$, $\ket{2,0}\rightarrow\ket{11}$, $\ket{2,1}\rightarrow\ket{9}$, $\ket{2,2}\rightarrow\ket{8}$, $\ket{2,3}\rightarrow\ket{5}$, $\ket{3,0}\rightarrow\ket{12}$, $\ket{3,1}\rightarrow\ket{13}$, $\ket{3,2}\rightarrow\ket{14}$, $\ket{3,3}\rightarrow\ket{15}$. Then we rewrite the set of states $\{\ket{\psi_k}\}_{k=1}^{54}$ in $ 4\otimes 4\otimes 4$ as $\{\ket{\varphi_k}\}_{k=1}^{54}$ in $ 4\otimes {16}$,
	\begin{equation}\label{eq:tile416}
	\begin{aligned}
	\ket{\varphi_{1,2,3}}&=\ket{1,0}+w_3^s\ket{2,1}+w_3^{2s}\ket{3,2},    \\ \ket{\varphi_{4,5,6}}&=\ket{1,1}+w_3^s\ket{2,2}+w_3^{2s}\ket{3,0},   \\
	\ket{\varphi_{7,8,9}}&=\ket{1,2}+w_3^s\ket{2,0}+w_3^{2s}\ket{3,1},  \\
	\ket{\varphi_{10,11,12}}&=\ket{1,3}+w_3^s\ket{2,4}+w_3^{2s}\ket{3,5},   \\
	\ket{\varphi_{13,14,15}}&=\ket{1,4}+w_3^s\ket{2,5}+w_3^{2s}\ket{3,3}, \\
	\ket{\varphi_{16,17,18}}&=\ket{1,5}+w_3^s\ket{2,3}+w_3^{2s}\ket{3,4}, \\
	\ket{\varphi_{19,20,21}}&=\ket{3,10}+w_3^s\ket{3,9}+w_3^{2s}\ket{3,14},\\  
	\ket{\varphi_{22,23,24}}&=\ket{3,7}+w_3^s\ket{3,8}+w_3^{2s}\ket{3,12},     \\
	\ket{\varphi_{25,26,27}}&=\ket{3,6}+w_3^s\ket{3,11}+w_3^{2s}\ket{3,13},\\
	\ket{\varphi_{28,29,30}}&=\ket{0,10}+w_3^s\ket{1,11}+w_3^{2s}\ket{2,12},  \\
	\ket{\varphi_{31,32,33}}&=\ket{0,11}+w_3^s\ket{1,12}+w_3^{2s}\ket{2,10}, \\
	\ket{\varphi_{34,35,36}}&=\ket{0,12}+w_3^s\ket{1,10}+w_3^{2s}\ket{2,11},  \\
	\ket{\varphi_{37,38,39}}&=\ket{0,13}+w_3^s\ket{1,14}+w_3^{2s}\ket{2,15},  \\
	\ket{\varphi_{40,41,42}}&=\ket{0,14}+w_3^s\ket{1,15}+w_3^{2s}\ket{2,13},  \\
	\ket{\varphi_{43,44,45}}&=\ket{0,15}+w_3^s\ket{1,13}+w_3^{2s}\ket{2,14}, \\ 
	\ket{\varphi_{46,47,48}}&=\ket{0,1}+w_3^s\ket{0,6}+w_3^{2s}\ket{0,5},    \\
	\ket{\varphi_{49,50,51}}&=\ket{0,2}+w_3^s\ket{0,4}+w_3^{2s}\ket{0,9},   \\
	\ket{\varphi_{52,53,54}}&=\ket{0,3}+w_3^s\ket{0,7}+w_3^{2s}\ket{0,8},     
	\end{aligned}
\end{equation}
	where $s=0,1,2$. Eq.~(\ref{eq:tile416})  corresponds to the $4\times 16$ grid in Fig.~\ref{Figure:tite416}. For example, $\ket{\varphi_{1,2,3}}$ corresponds to the cell set $\{(1,0),(2,1),(3,2)\}$. 	We need to show that  $\{\ket{\varphi_k}\}_{k=1}^{54}$ given by Eq.~(\ref{eq:tile416}) is locally irreducible.

\begin{figure*}[t]
	\centering
	\includegraphics[scale=0.6]{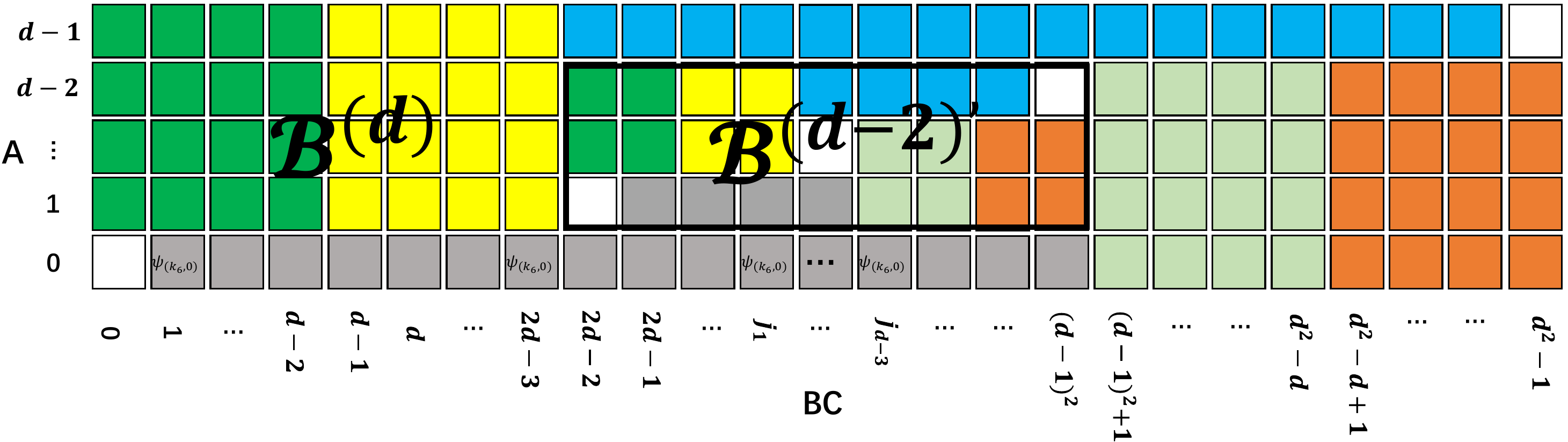}
	\caption{The corresponding $d\times d^2$ grid of $\cB^{(d)}$ given by Eq.~(\ref{eq:tileddd}) across $A|BC$ bipartition.  Define a bijection, $\ket{j}\rightarrow\ket{j+1}$ for $0\leq j\leq d-3$, then $\cB^{(d-2)}$  given by Eq.~(\ref{eq:tileddd}) is transferred into $\cB^{(d-2)'}$, and $\cB^{(d-2)'}$ across $A|BC$ bipartition corresponds to the inside $(d-2)\times (d-2)^2$ grid. Note that $\{\ket{\psi_{(k_6,0)}}\}_{k_6=0}^{d-2}$  given by Eq.~(\ref{eq:tileddd}) across $A|BC$ bipartition  corresponds to the $(d-1)$ cells in $\{(0,1),(0,j_1),\ldots, (0,j_{d-3}),(0,2d-3)\}$, where $2d-2\leq j_1\neq \cdots \neq j_{d-3}\leq (d-1)^2$. } \label{Figure:ddsquare}
\end{figure*}

	Let Alice go first and start the orthogonality-preserving POVM, $E_m=M_1^\dagger M_1=(a_{\ell,j})_{0\leq \ell,j\leq 3}$. Then the states of $\{M_1\otimes I\ket{\varphi_{k}}\}_{k=1}^{54}$ are mutually orthogonal. In order to show that the off-diagonal elements of $E_m$ are all zero, we need to choose the cells with same column index. For the same discussion as Lemma~\ref{lem:strong}, we obtain $a_{0,1}=a_{0,2}=a_{0,3}=0$ by  $(0,1)$ and $\{(1,1),(2,1),(3,1)\}$. If we choose $\{(1,0),(2,1),(3,2)\}$ and $\{(2,0),(3,1),(1,2)\}$, we can obtain $a_{1,2}=a_{2,3}=a_{3,1}=0$ by Lemma~\ref{lem:varder}. Thus the off-diagonal elements of $E_m$ are all zero.  For diagonal elements of $E_m$, we choose $\{(1,0),(2,1),(3,2)\}$. Then we obtain $a_{1,1}=a_{2,2}=a_{3,3}$ by Lemma~\ref{lem:varder}. We can also obtain $a_{0,0}=a_{1,1}=a_{2,2}$ by $\{(0,10),(1,11),(2,12)\}$. Thus the diagonal elements of $E_m$ are all equal. It means that $E_m$ is proportional to the identity matrix, and Alice cannot go first.
	
	Let Bob and Charlie go first and start the orthogonality-preserving POVM,  $E_m'=M_2^\dagger M_2=(b_{\ell,j})_{0\leq \ell,j\leq 15}$.  Then the states $\{I\otimes M_2\ket{\varphi_{k}}\}_{k=1}^{54}$ are mutually orthogonal.  Since Fig.~\ref{Figure:tite416} is centrosymmetric,  $b_{\ell,j}=0$ can implies that $b_{15-\ell,15-j}=0$ for any $\ell\neq j$. Next, since Fig.~\ref{Figure:tite416} has the similar structure as Fig.~\ref{Figure:tite39}, we obtain $b_{\ell,j}=0$ for $0\leq \ell\leq 5$ and $\ell+1\leq j\leq 15$ by the same discussion as Lemma~\ref{lem:strong}. We only need to  calculate $b_{6,7},b_{6,8},b_{6,9},b_{7,8}$. In Fig.~\ref{Figure:tite416}, we can choose $\{(0,1),(0,5),(0,6)\}$ and $\{(0,3),(0,7),(0,8)\}$. Since $b_{1,3}=b_{1,7}=b_{1,8}=b_{5,3}=b_{5,7}=b_{5,8}=b_{6,3}=0$, we obtain $b_{6,7}=b_{6,8}=0$ by Lemma~\ref{lem:varder}. We can also obtain $b_{6,9}=0$ by $\{(0,1),(0,5),(0,6)\}$ and $\{(0,2),(0,4),(0,9)\}$. Further, since  $\bra{\psi_{52}}I\otimes E_m'\ket{\psi_{53}}=\bra{\psi_{53}}I\otimes E_m'\ket{\psi_{54}}=\bra{\psi_{52}}I\otimes E_m'\ket{\psi_{54}}=\bra{\psi_{53}}I\otimes E_m'\ket{\psi_{52}}=0$, we have
	\begin{equation}
	\left\{\begin{split}
	&a_{3,3}+w_3a_{7,7}+w_3^2a_{7,8}+w_3a_{8,7}+w_3^2a_{8,8}=0,\\
	&a_{3,3}+w_3a_{7,7}+a_{7,8}+a_{8,7}+w_3^2a_{8,8}=0,\\
	&a_{3,3}+w_3^2a_{7,7}+w_3a_{7,8}+w_3^2a_{8,7}+w_3a_{8,8}=0,\\
	&a_{3,3}+w_3^2a_{7,7}+w_3^2a_{7,8}+w_3a_{8,7}+w_3a_{8,8}=0.
	\end{split}
	\right.
	\end{equation}
	Then, we have $w_3^2a_{7,8}+w_3a_{8,7}=a_{7,8}+a_{8,7}$, and $w_3a_{7,8}+w_3^2a_{8,7}=w_3^2a_{7,8}+w_3a_{8,7}$.  It implies $a_{7,8}=a_{8,7}=0$. Thus  the off-diagonal elements of $E_m'$ are all zero.  For diagonal elements of $E_m'$, if we choose  $\{(1,0),(2,1),(3,2)\}$,  it implies $b_{0,0}=b_{1,1}=b_{2,2}$ by Lemma~\ref{lem:varder}.  In the same way, we obtain $b_{3,3}=b_{4,4}=b_{5,5}$ by $\{(1,3),(2,4),(3,5)\}$, $b_{1,1}=b_{5,5}=b_{6,6}$ by $\{(0,1),(0,5),(0,6)\}$, and $b_{3,3}=b_{7,7}=b_{8,8}$ by $\{(0,3),(0,7),(0,8)\}$. It implies $b_{0,0}=b_{1,1}=b_{2,2}=b_{3,3}=b_{4,4}=b_{5,5}=b_{6,6}=b_{7,7}=b_{8,8}$.  Since Fig.~\ref{Figure:tite416} is centrosymmetric, we also obtain $b_{7,7}=b_{8,8}=b_{9,9}=b_{10,10}=b_{11,11}=b_{12,12}=b_{13,13}=b_{14,14}=b_{15,15}$. Thus the diagonal elements of $E_m'$ are all equal. It means that $E_m'$ is proportional to the identity matrix, and Bob and Charlie cannot go first.
	
	Thus, the orthogonal entangled set $\{\ket{\psi_k}\}_{k=1}^{54}$ given by Eq.~(\ref{eq:tile444}) is strongly nonlocal. 	
\end{proof}

\section{The proof of Theorem~\ref{thm:generalddd}}\label{appendix:lemddd}	

\begin{proof}
	(i) We prove it by induction on $d$.  Assume  $\cB^{(d-2)}$ given by Eq.~(\ref{eq:tileddd}) is strongly nonlocal when $d\geq 5$. Define a bijection, $\ket{j}\rightarrow\ket{j+1}$ for $0\leq j\leq d-3$, then $\cB^{(d-2)}$ is mapped to a set $B^{(d-2)'}\subset B^{(d)}$ in $d\otimes d \otimes d$. Obviously, $\cB^{(d-2)'}$ is also strongly nonlocal.  Consider the $A|BC$ bipartition of $B^{(d)}$ in Eq.~(\ref{eq:tileddd}), then it corresponds to the $d\times d^2$ grid in Fig.~\ref{Figure:ddsquare} by some permutations and bijections. Since  $\cB^{(d-2)'}\subset\cB^{(d)}$, then $\cB^{(d-2)'}$ across $A|BC$ bipartition corresponds to the  $(d-2)\times (d-2)^2$ grid in  Fig.~\ref{Figure:ddsquare}. For example, if $d=5$, then $\cB^{(3)'}$ corresponds to the $3\times 9$ grid in Fig.~\ref{Figure:ddsquare} (see also Fig.~\ref{Figure:tite39}). For the same discussion as Lemma~\ref{lem:strong} and Lemma~\ref{lem:strong444}, we can show that Alice cannot go first. Let Bob and Charlie go first and start the orthogonality-preserving POVM,  $E_m'=M_2^\dagger M_2=(b_{\ell,j})_{0\leq \ell,j\leq d^2-1}$. By
	the induction hypothesis, we can obtain  $b_{\ell,j}=0$ for $2d-2\leq \ell\leq (d-1)^2-1$ and $\ell+1\leq j\leq (d-1)^2$, and $b_{2d-2,2d-2}=b_{2d-1,2d-1}=\cdots=b_{(d-1)^2,(d-1)^2}$. By the same discussion as Lemma~\ref{lem:strong}, we can show that the off-diagonal elements of $E_m'$ are all zero. For the diagonal elements of $E_m'$,   we know that $b_{0,0}=b_{1,1}\cdots= b_{d-2,d-2}$, and $b_{d-1,d-1}=b_{d,d}\cdots=b_{2d-3,2d-3}$ by using $\{(1,0),(2,1),\ldots,(d-1,d-2)\}$ and $\{(1,d-1),(2,d),\ldots,(d-1,2d-3)\}$ in Fig.~\ref{Figure:ddsquare}  and Lemma~\ref{lem:varder}.  Further, the $(d-1)$ states $\{\ket{\psi_{(k_6,0)}}\}_{k_6=0}^{d-2}$  given by Eq.~(\ref{eq:tileddd}) across $A|BC$ bipartition  correspond to the $(d-1)$ cells in $\{(0,1),(0,j_1),\ldots, (0,j_{d-3}),(0,2d-3)\}$  in Fig.~\ref{Figure:ddsquare}, where $2d-2\leq j_1\neq \cdots \neq j_{d-3}\leq (d-1)^2$. Then we obtain  $b_{1,1}=b_{j_1,j_1}=\cdots=b_{j_{d-3},j_{d-3}}=b_{2d-3,2d-3}$ by Lemma~\ref{lem:varder}. It implies $b_{1,1}=b_{2,2}=\cdots=b_{(d-1)^2,(d-1)^2}$.  Since Fig.~\ref{Figure:ddsquare} is centrosymmetric, the diagonal elements of $E_m'$ are all equal.  It means that $E_m'$ is proportional to the identity matrix, and hence Bob and Charlie cannot go first. We obtain that $\cB^{(d)}$ is also strongly nonlocal.
	
	(ii) The proof is the same as (i).
\end{proof}

\section{The proof of Proposition~\ref{pro:distinguish333}}\label{appendix:dis333}
\begin{proof}
	First, two $\ket{\phi(2)}$ are distributed between Alice and Bob, two $\ket{\phi(2)}$ are distributed  Alice and Charlie, and one $\ket{\phi(2)}$ are distributed between Bob and Charlie. The initial state is
	\begin{equation}
	\begin{aligned}
	&\ket{\psi}_{A,B,C}\otimes\ket{\phi(2)}_{a_1,b_1}\otimes\ket{\phi(2)}_{a_2,c_1}\\
	&\otimes\ket{\phi(2)}_{a_3,c_2}\otimes\ket{\phi(2)}_{b_2,c_3}\otimes\ket{\phi(2)}_{a_4,b_3},
	\end{aligned}
	\end{equation}
	where $a_1$, $a_2$, $a_3$, and $a_4$  are the ancillary systems of Alice, $b_1$, $b_2$ and $b_3$ are the ancillary systems  of Bob, and $c_1$, $c_2$ and $c_3$ are the ancillary systems  of Charlie. Now the discrimination protocol proceeds as follows.
	
	\textit{Step 1}.   Bob performs the measurement $\{M_1:=P[\ket{0}_B;\ket{0}_{b_1}]+P[(\ket{1},\ket{2})_B;\ket{1}_{b_1}], \overline{M_1}:=I-M_1\}$. Charlie performs the measurement $\{M_2:=P[(\ket{0},\ket{1})_C;\ket{0}_{c_1}]+P[\ket{2}_C;\ket{1}_{c_1}], \overline{M_2}:=I-M_2\}$. If $M_1$ and $M_2$ click, the resulting postmeasurement states are
	\begin{widetext}
		\begin{equation*}
		\begin{aligned}
		\ket{\psi_{1,2}}&\rightarrow(\ket{1,0,0}\pm\ket{2,0,1})_{A,B,C}\ket{0,0}_{a_1,b_1}\ket{0,0}_{a_2,c_1}\ket{\phi(2)}_{a_3,c_2}\ket{\phi(2)}_{b_2,c_3}\ket{\phi(2)}_{a_4,b_3},  \\  \ket{\psi_{3,4}}&\rightarrow(\ket{1,0,1}\pm\ket{2,0,0})_{A,B,C}\ket{0,0}_{a_1,b_1}\ket{0,0}_{a_2,c_1}\ket{\phi(2)}_{a_3,c_2}\ket{\phi(2)}_{b_2,c_3}\ket{\phi(2)}_{a_4,b_3},    \\
		\ket{\psi_{5,6}}&\rightarrow(\ket{1,0,2}_{A,B,C}\ket{0,0}_{a_1,b_1}\ket{1,1}_{a_2,c_1}\pm\ket{2,1,2}_{A,B,C}\ket{1,1}_{a_1,b_1}\ket{1,1}_{a_2,c_1})\ket{\phi(2)}_{a_3,c_2}\ket{\phi(2)}_{b_2,c_3}\ket{\phi(2)}_{a_4,b_3},    \\
		\ket{\psi_{7,8}}&\rightarrow(\ket{1,1,2}_{A,B,C}\ket{1,1}_{a_1,b_1}\ket{1,1}_{a_2,c_1}\pm\ket{2,0,2}_{A,B,C}\ket{0,0}_{a_1,b_1}\ket{1,1}_{a_2,c_1})\ket{\phi(2)}_{a_3,c_2}\ket{\phi(2)}_{b_2,c_3}\ket{\phi(2)}_{a_4,b_3},  \\
		\ket{\psi_{9,10}}&\rightarrow(\ket{2,1,0}\pm\ket{2,2,1})_{A,B,C}\ket{1,1}_{a_1,b_1}\ket{0,0}_{a_2,c_1}\ket{\phi(2)}_{a_3,c_2}\ket{\phi(2)}_{b_2,c_3}\ket{\phi(2)}_{a_4,b_3},   \\
		\ket{\psi_{11,12}}&\rightarrow(\ket{2,1,1}\pm\ket{2,2,0})_{A,B,C}\ket{1,1}_{a_1,b_1}\ket{0,0}_{a_2,c_1}\ket{\phi(2)}_{a_3,c_2}\ket{\phi(2)}_{b_2,c_3}\ket{\phi(2)}_{a_4,b_3},    \\
		\ket{\psi_{13,14}}&\rightarrow(\ket{0,1,0}\pm\ket{1,2,0})_{A,B,C}\ket{1,1}_{a_1,b_1}\ket{0,0}_{a_2,c_1}\ket{\phi(2)}_{a_3,c_2}\ket{\phi(2)}_{b_2,c_3}\ket{\phi(2)}_{a_4,b_3},   \\
		\ket{\psi_{15,16}}&\rightarrow(\ket{0,2,0}\pm\ket{1,1,0})_{A,B,C}\ket{1,1}_{a_1,b_1}\ket{0,0}_{a_2,c_1}\ket{\phi(2)}_{a_3,c_2}\ket{\phi(2)}_{b_2,c_3}\ket{\phi(2)}_{a_4,b_3}.\\
		\ket{\psi_{17,18}}&\rightarrow(\ket{0,2,1}_{A,B,C}\ket{1,1}_{a_1,b_1}\ket{0,0}_{a_2,c_1}\pm\ket{1,2,2}_{A,B,C}\ket{1,1}_{a_1,b_1}\ket{1,1}_{a_2,c_1})\ket{\phi(2)}_{a_3,c_2}\ket{\phi(2)}_{b_2,c_3}\ket{\phi(2)}_{a_4,b_3},  
		\end{aligned}
		\end{equation*}
		\begin{align}
		\notag
		\ket{\psi_{19,20}}&\rightarrow(\ket{0,2,2}_{A,B,C}\ket{1,1}_{a_1,b_1}\ket{1,1}_{a_2,c_1}\pm\ket{1,2,1}_{A,B,C}\ket{1,1}_{a_1,b_1}\ket{0,0}_{a_2,c_1})\ket{\phi(2)}_{a_3,c_2}\ket{\phi(2)}_{b_2,c_3}\ket{\phi(2)}_{a_4,b_3},\\
		\notag
		\ket{\psi_{21,22}}&\rightarrow(\ket{0,0,1}_{A,B,C}\ket{0,0}_{a_1,b_1}\ket{0,0}_{a_2,c_1}\pm\ket{0,1,2}_{A,B,C}\ket{1,1}_{a_1,b_1}\ket{1,1}_{a_2,c_1})\ket{\phi(2)}_{a_3,c_2}\ket{\phi(2)}_{b_2,c_3}\ket{\phi(2)}_{a_4,b_3},  \\
		\ket{\psi_{23,24}}&\rightarrow(\ket{0,0,2}_{A,B,C}\ket{0,0}_{a_1,b_1}\ket{1,1}_{a_2,c_1}\pm\ket{0,1,1}_{A,B,C}\ket{1,1}_{a_1,b_1}\ket{0,0}_{a_2,c_1})\ket{\phi(2)}_{a_3,c_2}\ket{\phi(2)}_{b_2,c_3}\ket{\phi(2)}_{a_4,b_3}.
		\end{align}
	\end{widetext}

	\textit{Step 2}.  Alice performs the measurement $\{M_{2,1}:=P[(\ket{1},\ket{2})_A;\ket{0}_{a_1};\ket{0}_{a_2}], M_{2,2}:=P[\ket{2}_A;\ket{1}_{a_1};\ket{0}_{a_2}], \overline{M_2}:=I-M_{2,1}-M_{2,2}\}$. If $M_{2,1}$ clicks, it remains $\ket{\psi_{1,2,3,4}}$. These four states can be locally distinguished by using $\ket{\phi(2)}_{a_3,c_2}$. If $M_{2,2}$ clicks, it remains $\ket{\psi_{9,10,11,12}}$. These four states can also be locally distinguished by using $\ket{\phi(2)}_{b_2,c_3}$; if $\overline{M_2}$ clicks, it remains $\{\ket{\psi_k}_{k=5}^{8}\bigcup\ket{\psi_k}_{k=13}^{24}\}$.
	
	\textit{Step 3}.  Charlie performs the measurement $\{M_{3}:=P[\ket{0}_C],  \overline{M_3}:=I-M_3\}$. If $M_{3}$ clicks, it remains $\ket{\psi_{13,14,15,16}}$. These four states can be locally distinguished by using $\ket{\phi(2)}_{a_4,b_3}$; if $\overline{M_3}$ clicks, it remains $\{\ket{\psi_k}_{k=5}^{8}\bigcup\ket{\psi_k}_{k=17}^{24}\}$。
	
	\textit{Step 4}.  Bob performs the measurement $\{M_{4}:=P[\ket{2}_B],  \overline{M_4}:=I-M_4\}$. If $M_{4}$ clicks, it remains $\ket{\psi_{17,18,19,20}}$. These four states can be locally distinguished by the similar protocol as Example~\ref{example:belldisting}; if $\overline{M_4}$ clicks, it remains $\{\ket{\psi_k}_{k=5}^{8}\bigcup\ket{\psi_k}_{k=21}^{24}\}$.
	
	\textit{Step 5}.  Alice performs the measurement $\{M_{5,1}:=P[\ket{1}_A;\ket{0}_{a_1};\ket{1}_{a_2}]+P[\ket{2}_A;\ket{1}_{a_1};\ket{1}_{a_2}], M_{5,2}:=P[\ket{1}_A;\ket{1}_{a_1};\ket{1}_{a_2}]+P[\ket{2}_A;\ket{0}_{a_1};\ket{1}_{a_2}],
	M_{5,3}:=P[\ket{0}_A;\ket{0}_{a_1};\ket{0}_{a_2}]+P[\ket{0}_A;\ket{1}_{a_1};\ket{1}_{a_2}], \overline{M_5}:=I-M_{5,1}-M_{5,2}-M_{5,3}\}$. If $M_{5,1}$ clicks, it remains $\ket{\psi_{5,6}}$; if $M_{5,2}$ clicks, it remains $\ket{\psi_{7,8}}$; if $M_{5,3}$ clicks, it remains $\ket{\psi_{21,22}}$;   if $\overline{M_5}$ clicks, it remains $\ket{\psi_{23,24}}$.
	
	All other cases in step 1 obtain a similar protocol. From the beginning to Step 5, it consumes $1+\frac{4}{24}=\frac{7}{6}$ $\ket{\phi(2)}$  between Alice and Bob, $1+\frac{4}{24}=\frac{7}{6}$  $\ket{\phi(2)}$ between Alice and Charlie, and  $\frac{4}{24}=\frac{1}{6}$  $\ket{\phi(2)}$  between Bob and Charlie averagely.
\end{proof}
\vspace{0.5cm}

\vspace{0.4cm}	
\bibliographystyle{IEEEtran}
\bibliography{reference}
\end{document}